\documentclass[10pt,conference,letterpaper]{IEEEtran}

\usepackage{url}

\usepackage{amsmath,amssymb,amsfonts,amsthm}
\usepackage{graphicx}
\usepackage{array}
\usepackage{bbm}

\newtheorem{assumption}{Assumption}
\newtheorem{theorem}{Theorem}

\newtheorem{lemma}{Lemma}
\newtheorem{corollary}{Corollary}
\newtheorem{remark}{Remark} 
\newtheorem{definition}{Definition}  
    
\usepackage{float}
\usepackage{tikz}
\usetikzlibrary{decorations.pathreplacing}

\ifodd 0
\newcommand{\rev}[1]{{\color{blue}#1}} 

\else
\newcommand{\rev}[1]{#1}
\fi

\begin{document}

\title{Paid Prioritization with Content Competition}

\author{\IEEEauthorblockN{Parinaz Naghizadeh}
\IEEEauthorblockA{\textit{Ohio State University}}
\and
\IEEEauthorblockN{Carlee Joe-Wong}
\IEEEauthorblockA{\textit{Carnegie Melon University}}
\and
\IEEEauthorblockN{Mung Chiang}
\IEEEauthorblockA{\textit{Purdue University}}
}

\maketitle

\begin{abstract}
We study the effects of allowing paid prioritization arrangements in a market with content provider (CP) competition. We consider competing CPs who pay prioritization fees to a monopolistic ISP so as to offset the ISP's cost for investing in infrastructure to support fast lanes. Unlike prior works, our proposed model of users' content consumption accounts for multi-purchasing (i.e., users simultaneously subscribing to more than one CP). \rev{This model allows us to account for the ``attention'' received by each CP, and consequently to draw a contrast between how subscription-revenues and ad-revenues are impacted by paid prioritization.} We show that there exist incentives for the ISP to build additional fast lanes subsidized by CPs with sufficiently high revenue (from either subscription fees or advertisements). 
We show that non-prioritized content providers need not lose users, yet may lose revenue from advertisements due to decreased attention from users. We further show that users will consume a wider variety of content in a prioritized regime, and that they can attain higher welfare provided that non-prioritized traffic is not throttled. We discuss some policy and practical implications of these findings \rev{and numerically validate them}. 
\end{abstract}


\section{Introduction}\label{sec:intro}

Network (net) neutrality, or the principle that Internet service should be agnostic to the content of its traffic, has been a controversial part of Internet regulatory policies for several years. In the United States, much of the legal debate on net neutrality centers around the Federal Communication Commission (FCC)'s authority to impose such regulations on Internet Service Providers (ISPs). 
Until 2017, the feasibility of offering many policies, including paid prioritization, or allowing content providers (CPs) to pay ISPs in order to guarantee better quality-of-service (QoS) for their users~\cite{npr-repeal}, was widely debated. However, while net neutrality regulations remain in some parts of the world~\cite{india-netneutrality}, they expired in 2018 in the U.S., on the grounds that the FCC does not have the required authority to impose them on ISPs~\cite{npr-repeal}. \rev{The ongoing deployment of fifth-generation (5G) wireless networks, which provide differentiated QoS based on individual application needs, may reignite these debates by opening the door to payment-based QoS differentiation, a version of paid prioritization that arguably violates some interpretations of net neutrality regulations~\cite{yoo20195g}.}

\subsection{Network Neutrality Debates}

One of the key rationales outlined in the FCC's plans for their 2018 ruling was the lack of incentives by ISPs to invest in new infrastructure given existing regulations~\cite{fcc-newseum}. Indeed, most ISPs hailed the FCC ruling as paving the way for them to offer new services like paid prioritization, giving them new ways to conduct business beyond pricing users' access. In particular, several ISPs have taken the stance that content providers who are the ``big users of bandwidth'' should bear the costs of building the additional infrastructure needed to handle their increasing traffic~\cite{attpolicy,frontier-free-internet}. ISPs' added pricing power in light of the FCC's  rulings could therefore provide the incentives and capital that they require to invest in additional network infrastructure~\cite{forbes-netneutrality}.

While ISPs' stance and the FCC's policy arguments have been driven by the effect of regulations on the Internet infrastructure, allowing for paid prioritization has inevitably raised questions about its effects on content providers, and in particular, users' access to content. Most content providers have condemned the FCC ruling for allowing ISPs too much control over consumer choices; nonetheless, smaller competitors have been at the forefront of these oppositions, with larger providers being less vocal~\cite{guardian-netneutrality,netflix-big-enough}. 
For instance, the CEO of Netflix stated that, while an important issue, net neutrality was not a ``primary battle'' for them as they are ``big enough to get the deals [they] want''~\cite{netflix-big-enough}. 
Smaller CPs on the other hand, 
who cannot afford to pay for such 
prioritized service, could be put at a disadvantage if users gravitate to other content providers with better QoS. This effect would be particularly acute if ISPs drive down the QoS of non-prioritized services, effectively forcing CPs to pay for prioritization. Thus, we might expect that paid prioritization would have limited effect on larger CPs, but hurt smaller ones.

We \emph{analytically evaluate these arguments} for and against paid prioritization in a model with competing CPs, of varying size and revenue models, who enter paid prioritization agreements with an ISP. We find that (1) there indeed exist \emph{incentives for bigger CPs} (either those charging high enough subscription fees, or those making sufficient revenue from advertisements) to offset the ISP's costs in constructing additional infrastructure in order to prioritize their traffic. We show that (2) such arrangements can lead to an \emph{increase in users' welfare} \rev{not only due to improved QoS, but also due to the \emph{greater diversity in the content they consume}, since} (3) \emph{smaller CPs need not lose users}, as long as non-prioritized traffic maintains its current QoS level.  


\subsection{Our Contributions}

Several papers have examined the potential consequences of relaxing network neutrality regulations; see~\cite{greenstein2016net,schuett2010network} for surveys. 
Existing works have analyzed the effects of relaxing net neutrality regulations both in the absence of content provider competition over users  \cite{bourreau2015net,ma2011public,economides2012network}, as well as under competition between two perfectly substitutable content providers \cite{choi2015net,pil2010net,cheng2011debate,guo2017effects,kourandi2015net}. We extend this literature by contributing the following: 

{\bf A new model that captures both content provider competition as well as \emph{multi-purchasing} by users} (i.e., users subscribing to more than one of the competing providers). \rev{Such multi-purchasing is increasingly common in streaming markets; for example, a 2019 study showed that 25\% of Netflix iPhone users also use the recently introduced Disney$+$ app~\cite{vox-ott}, and 46\% of U.S. broadband users subscribe to more than one over-the-top streaming service~\cite{mobile-ott}.} 
Thus, we model users' heterogeneous intrinsic preferences for different content using a variant of the Hotelling model with multi-purchasing \cite{anderson2010hotelling,anderson2019importance}. 
We assume that users' 
decisions consist of selecting (at most) two CPs as well as their order of consumption. That is, users first choose to access a primary content provider (e.g., read the news from their preferred outlet, \rev{or browse for a movie to watch on Netflix}), and may further choose to access a secondary provider afterwards (e.g. later read headlines from a second source as well, \rev{or browse for other options on Disney+}). Our model accounts for the different utilities the users derive from a primary vs. secondary CP due to potential overlaps in their offered content.\footnote{Our analysis of users' choice of primary and secondary CPs 
is based on a first-order approximation of users' utility from multiple CPs and is adopted for tractability.} 

\rev{Intuitively, one might expect that the possibility of \linebreak multi-purchasing would blunt the impact of paid prioritization compared to past models that consider only single-purchasing by consumers: even if a CP did not have the revenue to purchase prioritized service, it might be able to maintain some of its revenue, and thus sustain its business, as consumers' secondary CP choice. It is less clear, however, how the resulting impact on CP and ISP revenue will affect ISPs' willingness to invest in infrastructure upgrades, or users' welfare. Our proposed model allows us to quantify these effects by
%
capturing} how CPs' revenues from different sources (namely, subscriptions and advertisement display) are affected by the ``attention'' they receive from their users. In particular, while all users pay subscription fees to a CP regardless of the order of access, a CP will attain lower advertisement revenue from its secondary users. This reduction in ad-revenue is due to users spending less time on secondary content, and/or potentially being exposed to the same advertisements on an earlier platform. Hence, users' click-through-rates drop as they reach later content. Our model allows us to evaluate the resulting changes in CPs' revenues.


{\bf Shedding new light on the potential positive/negative impacts of paid prioritization on CPs} (Corollaries~\ref{cor:cp-users-change} and~\ref{cor:cp-revenue}). We show that if non-prioritized traffic maintains the same QoS despite a competing CP's prioritization, then the non-prioritized CP can still attain \emph{the same} number of users. That being said, some of these users will designate the non-prioritized CP to be their secondary content of choice. Due to this lowered attention from users, non-prioritized CPs who are reliant on ad-revenue will realize less profit in a prioritized regime.  

We note that this finding is in contrast with prior works that had not accounted for users' multi-purchasing \cite{cheng2011debate,guo2017effects,pil2010net}. In particular, as users access only one of two competing CPs in these models, when one of the CPs is prioritized, its competitor will lose users (even if its QoS is not lowered); 
the competitor would then lose both advertisement and/or subscription revenue. Our analysis therefore provides new insights into how competing CPs are affected by paid prioritization depending on their revenue models and users' consumption behavior.  In particular, unlike what we might expect from CPs' stated arguments against paid prioritization, the resulting CP revenue is just as dependent on the \emph{type} of revenue as the size (in terms of revenue rates) of the CP.

\rev{These findings shed light on the potential impacts of paid prioritization on the different business models adopted by competitors in the  ``streaming wars''~\cite{the-verge-streaming-wars,ads-instead-of-fees}. Current providers' offerings include subscription based (e.g. Netflix, Disney+), ad-based (e.g. Tubi, the Roku Channel), and hybrid offerings (e.g. Peacock, Hulu, Quibi). The increasing number of (upcoming) ad-supported streaming services suggests that a switch to free ad-based or cheaper ad-supported models instead of pricey subscriptions may provide these services with a competitive advantage~\cite{ads-instead-of-fees,gold-rush,FC-majorbattle}. Our results highlight the potential effects of paid prioritization on such business models: they may make CPs more vulnerable to prioritized competitors.}

{\bf Analysis of user welfare} (Corollary~\ref{cor:user-welfare}). Lastly, we show that the described prioritization arrangements will lead to an increase in users' welfare. This is not only due to the improvement of QoS for the prioritized CPs' users, 
but also as more users will be multi-purchasing (and hence be exposed to a wider variety of content) when prioritization is allowed.  These may help to rebut CP fears that paid prioritization will stifle competition and inhibit the variety of content viewed by users. Figure~\ref{fig:results-illustration} illustrates our findings. 


\begin{figure}[t]
    \centering
    \includegraphics[width=0.8\columnwidth]{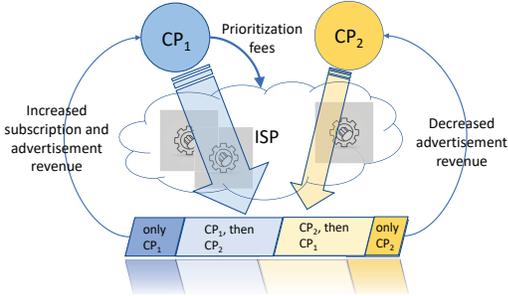}
    \caption{Changes in competing CPs' users and revenues following paid prioritization of $CP_1$ at the ISP in a 2-CP market. The top (solid) and bottom (shadowed) boxes show the users' content choices after and prior to prioritization, respectively. We show that when accounting for multi-purchasing of content, $CP_1$ gains users, while the competing $CP_2$ can maintain the same total number of users. Nonetheless, the non-prioritized $CP_2$ attains lower ad-revenue due to decreased attention from users.}
    \label{fig:results-illustration}
\end{figure}

\textbf{Policy and practical implications} (Section~\ref{sec:discussion}).  The above results heavily rely on a ``no throttling'' condition, i.e., that non-prioritized traffic is not slowed down. In particular, we show that in the absence of this assumption, non-prioritized CPs will lose users, and therefore their revenues from subscriptions will also decrease. We thus {\emph{validate CP fears that if ISPs drive down the QoS of non-prioritized traffic} (either on purpose, or due to restructuring of their resources), it could drive them out of business.} 
In addition, even though the variety of users' content consumption may still increase, depending on the extent of decrease in QoS of non-prioritized content, users' welfare may decrease as well. Regulators might then seek to impose this condition in order to ensure that paid prioritization benefits users and CPs as well as ISPs. {We elaborate on this and other implications of our findings in Section~\ref{sec:discussion}}. 



    
    
    

\textbf{Paper organization:} 
We review related work in Section~\ref{sec:related} before introducing our model in Section~\ref{sec:model}. We characterize the market equilibrium and discuss our main findings in Section~\ref{sec:analysis}. We discuss practical and policy implications of our findings, as well as potential model extensions and limitations in Section~\ref{sec:discussion}. We validate our results with numerical studies in Section~\ref{sec:simulations}, and conclude in Section~\ref{sec:conclusion}. All proofs are provided in the appendix.

\section{Related Work} \label{sec:related}

ISPs' incentives for capacity expansion have been studied with monopolistic  \cite{pil2010net,cheng2011debate,kramer2012network,ma2017paid} and competitive ISPs~\cite{bourreau2015net}. These works analyze ISPs' investment of  additional revenue from paid prioritization into expanding network capacity, and its effects on user welfare, in models with two substitutable CPs~\cite{pil2010net,cheng2011debate}, or in the absence of CP competition~\cite{kramer2012network,ma2017paid,bourreau2015net,tang2019regulating,ma2011public}. We complement this literature by considering both CP competition and users' multi-purchasing. 

Many works on paid prioritization do not consider CP competition. Some, in fact, bypass CPs by assuming that users choose their service classes depending on their CP preferences~\cite{odlyzko1999paris,zou2017optimal}. \rev{Other works instead bypass user models by assuming that CPs earn a fixed revenue per unit  capacity~\cite{nault2019balancing,tang2019regulating}.} Even models that allow CPs to pay ISPs for prioritization are generally considered in the context of non-competing CPs, e.g., if there is a continuum of CPs, and users choose their usage at each CP independently~\cite{economides2012economics,bourreau2015net,ma2011public}. 

The majority of existing work capturing CP competition~\cite{pil2010net,cheng2011debate,guo2017effects} has focused on two substitutable content providers, with competition modeled using the standard Hotelling model~\cite{hotelling1990stability} in which users have demands for one of the two CPs. An exception is the model of \cite{kourandi2015net} in which users can consume content from one or both providers. In \cite{kourandi2015net}, similar to our work, the CPs' ad-revenue depends on users' attention, with the revenue from secondary users attenuated. However, the focus of \cite{kourandi2015net} is on exclusive content distribution (hard fragmentation), with users' consuming all content available on their selected ISP. In contrast to these works, we explicitly model users' decisions to single vs. multi-purchase content using a variant of the Hotelling model. \rev{Our more refined model of user demand patterns allows us to find that a prioritized regime can incentivize users to subscribe to a greater variety of content.}

Our model is similar to the models of Hotelling with multi-purchasing adopted in \cite{anderson2010hotelling,anderson2019importance,athey2018impact}. The focus of these works is however on the pricing strategies of CPs, and on contrasting the value of single vs. multi-purchasing users to media platforms when obtaining  financing from advertisers. In contrast, we focus on paid prioritization arrangements and the ISP's investment decisions, and examine the implications for paid prioritization policies. In addition, we adopt a different model of CP distribution (see also Section~\ref{sec:extensions}). 

In terms of CPs' business models, the majority of existing work focuses on ad-revenue CPs~\cite{pil2010net,cheng2011debate,guo2017effects,kourandi2015net,bourreau2015net}. {Several of these works compare the benefits of prioritization to big vs. small CPs, with big CPs being those with higher ad-revenue rates.} 
{In contrast to these works,  given our model of single vs. multi-purchasing users,
we can distinguish between CPs' subscription and advertising revenues. This allows us to \rev{show that the type of CPs' business models, as well as CP size, determines the effects of paid prioritization.} 
} 


Finally, some models consider not paid prioritization but sponsored content, in which content providers subsidize users' costs of consuming data, but cannot influence the QoS 
users receive. Such sponsorship can lead to bias in favor of larger CPs~\cite{joe2018sponsoring}, even when the sponsorship indirectly affects user QoS through changing user demands~\cite{zhang2015sponsored}. \rev{Our work is complementary to this literature as well.}
\section{Model}\label{sec:model}

\subsection{The Stakeholders} \label{sec:stakeholders}
We consider an economy comprised of a monopolistic  Internet service provider, competitive content providers, and a continuum of users. {Throughout, we use she/her for the ISP, he/his for CPs, and it/its for users.} Table \ref{t:notation} summarizes the notation. 

{
\begin{table}[t]
\begin{small}
\centering
    \caption{Summary of notation}
    \label{t:notation}
\begin{tabular}{ | c | m{0.75\columnwidth} | } 
 \hline
 {Notation} & {Description}\\ 
 \hline
 $M$ & Number of CPs\\
 \hline
 $V$ & Users' base value from content consumption\\
 \hline
 $x$ & Users' position on the $[0,1]$ continuum \\
 \hline
 $m(x,j)$ & Distance of user at $x$ from $CP_j$ \\
 \hline
 $t$ & Transport/fit cost, determining users' intrinsic preference for content\\
 \hline
 $\theta$ & Users' residual benefit rate from secondary content\\ 
 \hline
 $F$ & Access fee charged to users by the ISP for connectivity\\ 
 \hline
 $S_j$ & Subscription fee charged to users by $CP_j$\\ 
 \hline
 $d_j$ & QoS/delay experienced by $CP_j$'s users\\ 
 \hline
 $d_0$ & Default QoS/delay for non-prioritized traffic\\
 \hline
 $p_j(d_j)$ & Prioritization price paid by $CP_j$ (per-user) to the ISP for attaining QoS/delay $d_j$\\
 \hline
 $C(d_j)$ & ISP's cost of building a fast lane with QoS/delay $d_j$\\ 
 \hline
 $z_j$ & The decision of $CP_j$ regarding prioritization, with $z_j=1$ (resp. $z_j=0$) for opting in (resp. not opting)\\
 \hline
 $n_j$ & Mass of $CP_j$'s users\\
 \hline
 $n_{jl}$ & Mass of users accessing $CP_j$ as their $l$-th content of choice ($n_{j1}$ and $n_{j2}$ for primary and secondary users, respectively)\\
 \hline
 $r_j$ & $CP_j$'s revenue rate from displaying advertisements\\ 
 \hline
 $\delta$ & Rate of reduction in ad-revenue from secondary users\\ 
 \hline
 $\lambda$ & Users' average traffic/demand rate\\ 
 \hline
 \end{tabular}
 \end{small}
\end{table}
}

\paragraph{Users} There is a continuum of users of mass 1 with heterogeneous preferences over the available content. We quantify this heterogeneity by assuming that users' content preference distribution follows (a variant of) the Hotelling model. \rev{Hotelling models are commonly adopted in the study of competitive markets in which users have heterogeneous preferences over the competing firms, including in prior work on net neutrality implications on content provider competition~\cite{guo2017effects,cheng2011debate,pil2010net}.} 

Formally, we assume users are distributed uniformly at random over the $[0,1]$ interval. Let $M$ CPs sit at equally distanced points throughout the line, with $CP_j, ~ j\in\{1, 2, \ldots, M\}$ sitting at $\frac{j-1}{M-1}$. Then, the (dis-)utility to a user at position $x$ from consuming content from $CP_j$ is proportional to the user's distance, $m(x,j):=|x-\frac{j-1}{M-1}|$ from the $CP$, and is given by $tm(x,j)$. Here, $t\geq 0$ is a unit fit/transportation cost, determining the users' sensitivity to their consumed content. 
Users multi-purchase, i.e., they may choose to consume content from more than one CP. 

In addition to these intrinsic preferences, users are sensitive to the quality-of-service experienced when consuming content from $CP_j$. {We let $d_j$ denote an implicit, monotonically increasing function of the actual delay or other QoS-related factors which may hinder users' content consumption. We do not explicitly model the mechanism through which QoS or delays are determined at the ISP (e.g. M/M/1 queues), {instead assuming the ISP has provided sufficient resources (if needed by building additional infrastructure) to serve the CP's users at a certain QoS}. We choose the notation/interpretation of delays $d_j$ to draw a parallel with existing models in the literature which have considered congestion delays as users' QoS  \cite{pil2010net,cheng2011debate,guo2017effects}.} 
Users will pay the ISP an access fee $F$ for connectivity, and further pay a subscription fee $S_j$ to their selected CP(s).

\paragraph{CPs} There are $M$ content providers\footnote{Although our model will assume users consume content from at most two CPs, we will differentiate between center and end-CPs (which can be interpreted as mainstream and specialized content, respectively). Therefore, our conclusions can be attained in a model of at least 3 CPs.}  who are competing for users' attention, and profiting from both advertisements and subscriptions. In particular, content provider $CP_j$ has a revenue rate $r_j$ \rev{per unit of user demand}, from displaying advertisements to each of his users. 
Higher rates indicate that the CP is more efficient at generating ad-revenue, due to, e.g., optimized ad-display algorithms. {CPs can further generate revenue from charging users a subscription fee $S_j\in \mathbb{R}_{\geq 0}$}. In a non-neutral Internet, CPs may be offered an option to pay the ISP for getting their content prioritized, as detailed shortly. In this paper, to understand the effects of such paid prioritization, we take subscription fees $S_j$ and ad-revenue rates $r_j$ to be fixed, and focus on CPs' prioritization decisions. 

\paragraph{ISP} We consider a monopolistic service provider. {This is motivated by a lack of competition in the US ISP market, with studies showing that just 30\% of US homes have access to more than one ISP that offers broadband speeds~\cite{pcmag-isp-competition}.} 
The ISP charges users a flat rate access fee $F\in \mathbb{R}_{\geq 0}$ for connectivity.\footnote{\rev{Note in particular that users' data consumption volume does not affect their access fees to the ISP, as for typical broadband or WiFi subscription plans. Cellular data plans may impose caps or throttling on ``heavy'' users, but these do not affect most users.
}} 
The ISP can also charge the CPs for prioritization of their content. Again, as our focus in this paper is on the effects of prioritization (and also, as we do not consider ISP competition), we take the fees $F$ to be fixed, and focus on the ISP's choice of prioritization offers as detailed below.

We assume the ISP offers a two-tiered prioritization scheme, with a ``fast lane'' and a ``default" option. Let $d_0$ denote the default QoS/delay offered by the ISP. In order to offer a fast lane with improved QoS/delay $d\leq d_0$, the ISP should invest in improving or building infrastructure. We let $nC(d)$ be the cost to the ISP for offering a fast-lane with QoS/delay $d$ which can serve $n$ users. We assume the cost function $C(\cdot)$ is decreasing, differentiable, and convex.\footnote{\rev{The infrastructure cost should account for  both the requested QoS/delay $d_j$, as well as the number of users $n_j(d_j)$ at equilibrium given $d_j$. For instance, consider reservation-based (PMP~\cite{odlyzko1999paris}) delays of the form $d_j={n_j}/{\Phi_j}$, where $\Phi_j$ is the network capacity dedicated to $CP_j$, and let $n'_j$ be $CP_j$'s users before prioritization. Then, to serve a prioritized $CP_j$, an additional network capacity of ${n_j(d_j)}/{d_j}-{n'_j}/{d_0}$ is needed.}} 
The ISP then offsets these costs through prioritization prices charged to the requesting CP(s). In particular, to provide fast-lane access with QoS/delay $d_j$ to his $n$ users, a $CP_j$ will be charged a prioritization fee of $np_j(d_j)$ by the ISP, with $p_j$ denoting a CP-specific, QoS-dependent, per-user prioritization rate.\footnote{\rev{We note that this form of fees allows the ISP to tailor her prioritization offerings and differentiate between CPs. This is akin to the current negotiations of CP-specific special peering arrangements. In our numerical simulations in Section~\ref{sec:simulations}, we show that such differentiation will ultimately improve users' welfare. This is because under uniform pricing, the ISP may choose to focus solely on extracting profit from the highest revenue CP in the market. This prevents other CPs' users from experiencing better QoS, despite the willingness of their providers to subsidize dedicated fast lanes.}}

\subsection{The Utility Functions} \label{sec:utility-functions}

\paragraph{Users} Consider a user placed at position $x$; recall that this position determines the user's intrinsic preference for each CP. We adopt a Hotelling model with multi-purchasing, similar to those proposed in \cite{anderson2010hotelling,anderson2019importance},\footnote{\rev{Our model extends the 2-CP model of \cite{anderson2010hotelling} to $M\geq 3$ CPs. The model of \cite{anderson2019importance} considers $M$ providers as well, but distributed on a Salop circle. In contrast, we consider the distribution on the unit interval to capture the spectrum of user preferences.}} to evaluate users' content consumption choices. In particular, 
let $CP_j$ and $CP_k$ be the primary and secondary content providers selected by a user; this means that the user first consumes content from $CP_j$ (e.g., reads news articles from its most preferred website), followed by content from $CP_k$ (e.g., later reads headlines from a second source). 
Then, a user at $x$ will derive a utility of 
$$U_x(j) = V- tm(x,j)-d_{j} -S_j~,$$ 
from its primary provider, $CP_j$, and a residual utility of 
$$\bar{U}_x(k) = \theta\left(V - tm(x,k)- d_{k}\right) -S_k~,$$ 
if it also accesses a secondary CP, $CP_k$.  
Here, $V$ denotes the base value of consuming any content, and $\theta<1$ determines the rate of reduction in content enjoyment beyond the users' primary choice. This is a first-order approximation; we assume the rate of reduction $\theta$ for the third most preferred content and beyond is sufficiently small, so that users ignore the utility derived from consumption of content beyond their primary and secondary CPs when making their consumption decisions.\footnote{Specifically, we can define residual benefit rates $\theta_p$, with $\theta_M<\cdots<\theta_2<1$, for a user's $p$-th CP of choice among the $M$ CPs. In this paper, similar to prior work on Hotelling models with multi-purchasing \cite{anderson2019importance}, we limit our analysis to primary and secondary CP choices for tractability.} 

The residual benefit rate $\theta$ models the assumption that users derive less base utility $V$ from providers beyond their primary content due to {spending less time on secondary content and/or the} potential overlaps with their primary CP's content. For instance, users would derive less utility from reading news headlines at a secondary CP after already reading stories on the same topics at a primary CP. Similarly, they might derive less utility from watching videos on a secondary CP if they have already watched videos at their primary CP. Knowing this, the users are also less sensitive to their intrinsic preference $m(x,k)$ and the service quality  $d_k$ for secondary content. Therefore, all three effects are attenuated for the users' secondary CP.  
Note also that the rate of reduction  does not affect the subscription fees, as these are sunk costs paid by the users regardless of the order in which they access content. 

The total utility of a user at $x$, who pays access fee $F$ to the ISP, from content consumption from only $CP_j$, or from a pair of primary and secondary providers $CP_j$ and $CP_k$, will be given by
\begin{align}
U_x(\{j\}) &= U_x(j) -F~,\notag\\
U_x(\{j,k\}) &= U_x(j) + \bar{U}_x(k) -F~,
\label{eq:user-utility}
\end{align}
respectively. {Under these utilities users are homogeneous in their base valuation, delay sensitivity, and rate of consumption of content from each CP, but differ in their interest in CPs' content}. \rev{We note that (standard) Hotelling models have also been used in prior work \cite{pil2010net,cheng2011debate,guo2017effects} to capture users' intrinsic preferences over two competing content providers, with users' utilities given by $U_x(\{j\})$ (i.e., under single-purchasing only). Our proposed model will match these prior models when setting $\theta=0$; we consider this special case in Section~\ref{sec:disc-MCPs}.} In the following section, we will evaluate users' optimal choice of content consumption under the more general multi-purchasing model (i.e., choice of content providers, as well as their order of consumption).

\paragraph{CPs} Let $n_{jl}$ be the mass of $CP_j$ users that access his content as their $l$-th content of choice, and $n_j =n_{j1}+n_{j2}$ be the mass of all $CP_j$ users. We refer to $n_{j1}$ and $n_{j2}$ as $CP_j$'s primary and secondary users, respectively. 
Let $z_{j}\in \{0,1\}$ denote whether $CP_j$ has opted for prioritization, and $d_j$ be his users' QoS/delay. Then, $CP_j$'s revenue is given by
\begin{align}
R_j(d_j) =  n_jS_j + \lambda (r_j(n_{j1}+\delta n_{j2}) - z_{j}n_{j}p_j(d_j))~.
\label{eq:cp-revenue}
\end{align}
This revenue consists of profit from subscription fees, advertising revenue from primary and secondary users, and paid prioritization fees (if any). Here, $\lambda$ is users' average traffic/demand rate, and $\delta\leq 1$ is a rate of reduction in ad-revenue from users who choose to access this CP as their secondary source. \rev{For example, users may spend less time on their secondary CP, and thus see and click on fewer ads; or the secondary CP may experience a lower click-through rate on the ads he displays, due to users' having already seen these ads on their primary CP. We abstract from modeling these specific effects, instead encapsulating them in the single parameter $\delta$.}


\paragraph{ISP} 
The ISP is a profit maximizer, with profit given by
\begin{align}
    \Pi = F + \lambda \sum_{j=1}^M z_{j} n_{j}(p_j(d_j) - C(d_{j}))~.
    \label{eq:isp-profit}
\end{align}
This profit consists of users' access fees, CPs' prioritization fees, and the costs of building additional infrastructure. 

\subsection{Market Equilibrium}\label{sec:eq-steps}
Throughout the next sections, we study the effects of allowing paid prioritization options on the market equilibrium. This equilibrium consists of the users' content choices, their experienced QoS, the CPs' choices regarding prioritization, and the prioritization fees set by the ISP. Our analysis is based on the following order of decisions:
\begin{itemize}
\item The ISP offers prioritization options to the CPs. This consists of an improved QoS and the prioritization price that would be charged for it. 
\item CPs choose whether to opt for prioritization. 
\item Users decide their content consumption (i.e., choice of CPs and their order of consumption) based on their expected QoS given CPs' choices of prioritization. 
\end{itemize}
To analyze the above market equilibrium, we use the standard approach of backward induction to find the subgame perfect  equilibrium of the resulting Stackelberg game. 
\section{Prioritization with Competing CPs}\label{sec:analysis}

In this section, we first present our main theorem on the characterization of the market equilibrium. 
We then discuss the implications of this theorem on the potential effects of paid prioritization on different stakeholders through a number of corollaries. 

\subsection{Market Equilibrium with  Competing CPs}
In evaluating the market equilibrium, we make the following assumptions on the problem parameters. 

\begin{assumption}\label{as:parameters}
We assume that the basic problem parameters $\{V,t,d_0,\theta,\{S_j\},F\}$ satisfy the following conditions:
\begin{enumerate}
    \item {[High base value from content]:} $V>d_{0}+\frac{t}{M-1}+F+\max_j S_j$~.
    \item {[Moderate residual benefit from secondary content]:} \linebreak $\frac{\max_j S_j}{V-d_0-\frac12(\max_j S_j +\frac{t}{M-1})} < \theta < \frac{\min_j S_j}{V- \frac{t}{M-1}}$.
\end{enumerate}
\end{assumption}
Intuitively, this assumption can be interpreted as follows. \rev{The first condition is similar to the common ``full coverage'' assumption adopted also in prior work \cite{guo2017effects,cheng2011debate,pil2010net}}.  Specifically, under part (1), any user can get a positive utility from consuming content from at least two of the CPs (regardless of QoS and fees). Therefore, any user in the population will be consuming at least one content, and will have at least two CP options to choose from. 
In addition, by assuming part (2), we ensure that the residual benefit from secondary content, $\theta$, is neither too low to preclude multi-purchasing (i.e., some users do consume multiple content), but not too high to make multi-purchasing trivial (i.e., there will exist some users who subscribe to only one content). \rev{This assumption is new compared to prior work, and is due to our introduction of potential dual-purchasing behavior by users. We detail the derivation of each of the above conditions in Appendix~\ref{app:assumption1}.} 

We would like to emphasize that the conditions set in  Assumption \ref{as:parameters} make our comparisons of the effects of introducing prioritization options broader. In particular, under this assumption, all potential single-purchasing and dual-purchasing content consumption patterns will have a non-zero mass of users. This will allow us to evaluate the effects of prioritization on the changes in each possible class of users.\footnote{\rev{We discuss some relaxations of these conditions in Appendix~\ref{app:assumption1-relaxed},
and show that our qualitative results continue to hold as long as some level of dual-purchasing exists at the non-prioritized outcome. We analyze the extreme case of $\theta=0$ (i.e. single-purchasing) in more detail in  Section~\ref{sec:disc-MCPs}.}} 

The following theorem characterizes the (unique) paid prioritization agreement offered by the ISP (including the QoS offered in the fast lane, and the price charged to a prioritized CP). It also determines the mass and composition of each CP's users. Note that the theorem differentiates between the two end CPs ($CP_1$ and $CP_M$) and other CPs, as the end CPs are only competing with one other content provider, while the middle CPs have two direct competitors. 

\begin{theorem}\label{thm:MCPs}
Assume the conditions in Assumption \ref{as:parameters} are satisfied. Let $\tau_j(d_j):=\frac{1}{t}\left(V-d_j-\frac{S_j}{\theta}\right)$. Then, 
\begin{enumerate}
\item A $CP_j$, $j\in[2,M-1]$, with QoS/delay $d_j$ will have a total of $n_j=2\tau_j(d_j)$ users, with $n_{1j}=\frac{1}{M-1}+\frac{d_{j+1}+d_{j-1}-2d_j}{2t}$ of them being primary users, and the remainder as secondary users. 
\item A $CP_j$, $j\in\{1,M\}$, with QoS/delay $d_j$ will have a total of $n_j=\tau_j(d_j)$ users, with $n_{1j}=\frac{1}{2}(\frac{1}{M-1}+\frac{d_{j+(-1)^{\scriptsize{\mathbbm{1}\{j=M\}}}}-d_j}{t})$ of them being primary users, and the remainder as secondary users. 
\end{enumerate}
Further, if $CP_j$ has a paid prioritization contract with the ISP,
\begin{enumerate}
\setcounter{enumi}{2}
    \item The QoS offered to $CP_j$'s users in the fast lane will be determined by
    \begin{align*}
        d^*_j :=&   \arg\min_{d_{j}\in(0, d_0]} ~
         \frac{S_j+\frac{1}{2}(1+\delta)\lambda r_j}{t}d_j + \lambda\tau_j(d_j)C(d_j).
    \end{align*}
    \item $CP_j$ will be charged a prioritization fee of
    \begin{align*}
    p_j(d_j^*) = \frac{S_j+\frac{1}{2}(1+\delta)\lambda r_j}{\lambda t}\cdot \frac{d_0-d_j^*}{\tau_{j}(d_j^*)}.
    \end{align*}
\end{enumerate}
\end{theorem}

\paragraph*{Proof} See Appendix~\ref{app:thm1-proof-inpaper}. \hfill\qedsymbol



\rev{
\begin{remark}
At the heart of our analysis and findings is the term $\tau_{j}(d_j)=\frac{1}{t}\left(V-d_j-\frac{S_j}{\theta}\right)$. Intuitively, any user who is at distance of $\tau_{j}(d_j)$ or less from $CP_j$ (as measured in the Hotelling model) will derive benefit from $CP_j$'s content as a secondary user. This threshold is decreasing in $d_j$ and $S_j$, indicating that users are less interested in a content when it has lower QoS or higher subscription fee. The threshold is further increasing in $\theta$, meaning that as the residual benefit from secondary content increases, more users adopt it. 

Our findings in Parts (1) and (2) of Theorem \ref{thm:MCPs} indicate that this threshold can in fact also determine the total number of users that a CP can attract when competing with each of his neighboring CPs. This is because we find that once users lose interest in a $CP_j$ as their secondary provider, their distance from $CP_j$ is ``far'', in the sense that they would also prefer a competing $CP_{j\pm1}$ to $CP_{j}$ as their primary provider. In other words, we show that the CPs have a better chance of gaining the users who are less intrinsically interested in them as secondary users, as such users could still continue primarily consuming their most preferred content. 
\end{remark}

}

\subsection{Interpretation and Implications}\label{sec:disc-MCPs}

We now discuss the implications of Theorem~\ref{thm:MCPs} through a number of corollaries; all proofs are given in the appendix.

\subsubsection{Changes in CPs' users}
We begin by evaluating the change in users' content consumption. The following lemma establishes users' distribution before prioritization. 

\begin{lemma}\label{lemma:cp-users}
In a non-prioritized regime, each $CP_j$, $j\in[2,M-1]$, has $n_j=2\tau_j(d_0)$ users, with $n_{1j}=\frac{1}{M-1}$ of them being primary users, and the remainder as secondary users. For an end $CP_j$, $j\in\{1,M\}$,  $n_j=\tau_j(d_0)$ and $n_{1j}=\frac12\frac{1}{M-1}$.  
\end{lemma}
 
This lemma follows directly from Theorem \ref{thm:MCPs} when substituting $d_j=d_0, \forall j$, in Parts (1) and (2). Note that without prioritization, all (non-end) CPs attain the same number of primary users, yet  those with smaller subscription fee $S_j$ (hence, larger $\tau_j$) attract more secondary users. 

The next corollary describes the changes in the users of a CP, as well as his competitors, following prioritization. 

\begin{corollary}\label{cor:cp-users-change}
Assume $CP_j$, $j\in[2,M-1]$, has a prioritization agreement with improved QoS/delay $d_j^*<d_0$, while his competitors maintain the default QoS/delay of $d_0$. Then,  
\begin{enumerate}
    \item The prioritized $CP_j$ gains $2\frac{d_0-d_j^*}{t}$ additional users. Half are additional primary users, while the other half are additional secondary users. 
    \item The competitors $CP_k,~k\in\{j-1,j+1\}$ will have the same total number of users. However, $\frac12\frac{d_0-d_j^*}{t}$ of $CP_k$'s primary users will switch their consumption pattern, relegating $CP_k$ to their secondary content of choice.  
\end{enumerate}
For an end $CP_j$, $j\in\{1,M\}$, the corollary holds by dropping the factor 2 in part (1).
\end{corollary}

\textbf{Interpretation:} By opting for prioritization, a $CP_j$ can gain both primary and secondary users. The gain in primary users comes from the users who had dual-purchased $CP_j$'s content as their secondary content of choice prior to prioritization. Given the improved QoS, these users now re-order their consumption, choosing $CP_j$ as their primary source. The remaining gain in users comes from those who were single-purchasing from one of $CP_j$'s competitors before prioritization. Given the increase in $CP_j$'s quality, these users will now have sufficient incentive to subscribe to $CP_j$ as well, consuming his content as a secondary source. {Lemma~\ref{lemma:cp-users} and Corollary~\ref{cor:cp-users-change} are illustrated in Figure~\ref{fig:users-consumption-change}.} 

{\begin{remark}\label{r:non-exclusive}
If one of the competitors of $CP_j$, say $CP_{j+1}$, is also prioritized to $d_{j+1}^*<d_0$, then $CP_j$ will still gain $2\frac{d_0-d_j^*}{t}$ users. However, only $\frac12(\frac{d_0-d_j^*}{t}+\max\{\frac{d_{j+1}^*-d_j^*}{t},0\})$ of them will be additional primary users (who might re-order their content consumption to prioritize $CP_j$), depending on which CP has attained a better prioritized QoS. This means that ISPs' differentiation of CPs when offering fast lanes will not affect their ability to gain users (given the no-throttling assumption), but  affects the additional attention they can get. 
\end{remark}}

\textbf{Takeaway:} A prioritized CP gains new users, and can further get additional attention from his existing users depending on his competitors' QoS. Non-prioritized CPs do not lose users, but lose the attention of some of their existing users.

\begin{figure}[t]
\begin{center}
\begin{tikzpicture}[scale=0.7]

\draw [,-] (0,0) -- (10,0);
\draw [,-] (0,0.2) -- (0,-0.2);
\draw [,-] (10,0.2) -- (10,-0.2);
\draw [,-] (5.3,0.1) -- (5.3,-0.1);
\draw [,-] (1.6,0.1) -- (1.6,-0.1);
\draw [,-] (8.6,0.1) -- (8.6,-0.1);
\draw [,-] (5,0.1) -- (5,-0.1);
\draw [,-] (8,0.1) -- (8,-0.1);

\draw[fill=blue!30] (0,0.2) rectangle (1.6,0.8)  node[align=center] at (0.8,0.5) {$\{CP_j\}$};
\draw[fill=blue!10] (1.6,0.2) rectangle (5,0.8) node[align=center] at (3.3,0.5) {$\{CP_j,CP_{j+1}\}$};
\draw[fill=yellow!20] (5,0.2) rectangle (8,0.8) node[align=center] at (6.5,0.5) {$\{CP_{j+1},CP_j\}$};
\draw[fill=yellow!50] (8,0.2) rectangle (10,0.8) node[align=center] at (9,0.5) {$\{CP_{j+1}\}$};

\draw[fill=blue!30] (0,-0.2) rectangle (1.6,-0.8) node[align=center] at (0.8,-0.5) {$\{CP_j\}$};
\draw[fill=blue!10] (1.6,-0.2) rectangle (5.3,-0.8) node[align=center] at (3.45,-0.5) {$\{CP_j,CP_{j+1}\}$};
\draw[fill=yellow!20] (5.3,-0.2) rectangle (8.6,-0.8) node[align=center] at (6.95,-0.5) {$\{CP_{j+1},CP_j\}$};
\draw[fill=yellow!50] (8.6,-0.2) rectangle (10,-0.8) node[align=center] at (9.3,-0.5) {\tiny{$\{CP_{j+1}\}$}};

\node[align=center] at (-0.8,.5) {Before:};
\node[align=center] at (-0.8,-.5) {After:};

\end{tikzpicture}
\end{center}
    \caption{Changes in users' content consumption between $CP_j$ and $CP_{j+1}$ when $CP_j$ is prioritized. The top/bottom boxes represent the mass of users before/after prioritization, and $\{CP_i, CP_k\}$ indicates the mass of users with primary $CP_i$ and secondary $CP_k$. Note that $CP_{j+1}$ maintains the same total number of users, but has fewer primary users.} 
    \label{fig:users-consumption-change}
\end{figure}
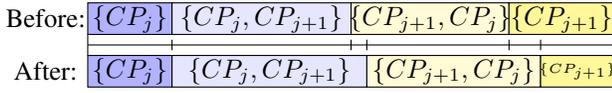

\subsubsection{Change in CPs' revenues} We now evaluate how the CPs' revenues are affected by prioritization decisions. 

\begin{corollary}\label{cor:cp-revenue}
Assume $CP_j$ establishes a prioritization contract with improved QoS/delay $d_j^*< d_0$, while his competitors maintain the default QoS/delay of $d_0$. Then, 
\begin{enumerate}
    \item The prioritized $CP_j$'s revenue is equal to the revenue prior to prioritization.
    \item The revenue of a (non-prioritized) competitor $CP_k$ ($k=j\pm1$) will decrease by $\frac{\lambda r_k (1-\delta)}{2t}(d_0-d_j^*)$. 
\end{enumerate}
\end{corollary}

\textbf{Interpretation:} First, we note that as the ISP is monopolistic, she  can choose to extract the entire surplus gained by a prioritized CP as a prioritization price, while still making the CP (weakly) incentivized to opt for prioritization. The loss in the competitor CP's revenue on the other hand is due to the reduction in primary users (as shown in Corollary \ref{cor:cp-users-change}). That is, although the total number of the competing CP's users does not change (hence, no change in subscription revenue), some users re-order their consumption to first access the prioritized CP, and therefore, reduce the competitor's ad-revenue. 

From Corollary \ref{cor:cp-revenue}, we observe that prioritization of a competitor will be most damaging to the revenue of content providers that are reliant on advertising. In particular, a CP making profit solely from  subscriptions would be indifferent to his competitors' prioritization, as, by Corollary \ref{cor:cp-users-change}, the total number of the CPs' users will remain unchanged. This conclusion is however dependent on our assumption that 
prioritization does not affect the QoS of non-prioritized content, as discussed in the following remark. 

\begin{remark}\label{r:throttling}
{By Theorem \ref{thm:MCPs}, the total number of a CP's users is given by $n_j=\tau_j(d_j)=\frac{1}{t}(V-d_j-\frac{S_j}{\theta})$. Hence, if $d_j$ is (purposefully) degraded, the CP will lose users. Further, we have assumed that separate fast lanes are constructed to handle prioritization (i.e., fixed $d_j=d_0$ for non-prioritized content). Nonetheless, prioritization might be achievable by restructuring existing ISP infrastructure. In particular, given a fixed capacity, fewer resources would be dedicated to non-prioritized traffic, thus reducing the QoS of non-prioritized content. This lowered QoS guarantees could reduce the total number of users of non-prioritized CPs at equilibrium, hence reducing their subscription-based revenues as well.}
\end{remark}

\textbf{Takeaway:} 
Despite not losing users, non-prioritized  CPs will have lower revenue due to reduced attention from their users, and hence, reduced advertising revenue. If the non-prioritized CPs' QoS decreases, they will lose revenue from subscriptions as well.  

\subsubsection{QoS attainable by each CP} We next determine how the subscription fee and revenue rate of a CP affect the QoS he can afford to attain for his users through prioritization. 

\begin{corollary}\label{cor:affordability}
Let $C(d_0)=0$ (i.e., there already exists infrastructure to serve the traffic at the standard QoS), and $\lim_{d\rightarrow 0} C(d) = \infty$ (negligible delays are prohibitively costly/infeasible). 
Then, 
\begin{enumerate}
    \item The users of $CP_j$ will get a fast lane with QoS/delay $d_j^*<d_0$ if and only if $|C'(d_0)|< \frac{S_j+\frac{1}{2}(1+\delta)\lambda r_j}{\lambda (V-d_0-\frac{S_j}{\theta})}$.
    \item If the above condition holds, the fast lane's QoS is increasing ($d_j^*$ is decreasing) in both $S_j$ and $r_j$. 
\end{enumerate} 
\end{corollary}

\textbf{Interpretation:}  As intuitively expected, Corollary \ref{cor:affordability} states that not all CPs can afford to compensate the ISP for the construction of fast lanes. Recall that the cost function $C(\cdot)$ is decreasing, with a larger $|C'(\cdot)|$ indicating a more steep increase in cost for improving QoS. Therefore, only CPs with sufficiently high revenue rates (either high $S_j$ or high $r_j$) may be able to afford to offset the costs of fast lane construction. 
Further, higher revenue CPs can afford to offset the ISP's cost for higher QoS fast lanes. 

It is also interesting to note that as $V$ increases, or as $d_0$ decreases (i.e., if the default QoS is relatively high), each CP has a higher mass of users even in the absence of prioritization, which in turns also decreases the CP's interest in offsetting the ISP's costs for fast lane construction. 
 
\textbf{Takeaway:} As intuitively expected, only CPs with sufficiently high subscription fees or advertising revenue may be able to pay for fast lanes. Among paying CPs, those with higher subscription fees or higher ad-revenue rates can negotiate higher QoS.

\subsubsection{Users' welfare} We next compare users' welfare under the assumption that non-prioritized traffic is not slowed down. 

\begin{corollary}\label{cor:user-welfare}
Assume $CP_j$ has a prioritization agreement with improved QoS/delay $d_j^*<d_0$, while his competitors maintain the default QoS/delay $d_0$. Then, the introduction of prioritization increases users' welfare. 
\end{corollary}

\textbf{Interpretation:} When content from $CP_j$ is prioritized, his users experience higher QoS. As such, the utility from any user who was already consuming content from this CP before prioritization, as well as that of his new secondary users, increases following prioritization. 
We also show that multi-purchasing users who switch their order of consumption in favor of the prioritized content attain a higher utility. Therefore, users with any consumption pattern attain higher utility than prior to the introduction of prioritization. 

We note that there are two driving factors in the improvement of user welfare: (1) users' QoS improves (which is also the driving factor of welfare improvement in prior work on paid prioritization), and (2) additional welfare is generated due to the increased variety of content consumption from the new secondary users of $CP_j$ (which is new to our model); in particular, these users continue to attain the same utility from consuming primary content from a competitor, while now also getting added utility from secondary content consumption. 

We note that the fact that the non-prioritized contents retain the same QoS as before is crucial for the conclusion of Corollary \ref{cor:user-welfare}. If the traffic of competitors is throttled (or congested) following $CP_j$'s prioritization, the utility of users consuming competitors' content would degrade. 
Depending on the extent of this degradation, overall user welfare may decrease. A similar conclusion has been made under other models of CP competition and user content consumption \cite{kramer2012network,tang2019regulating}. 

\textbf{Takeaway:} As long as prioritization is attained through building additional infrastructure (while maintaining the same QoS for non-prioritized content), users' welfare increases under prioritization due not only to improved QoS (as is also the case in prior work), but also due to increased variety of content consumption (which is new to our model). This result may not hold if non-prioritized contents is throttled or slowed down.

\rev{
\subsubsection{Single vs dual-purchasing} Finally, we characterize the market equilibrium for a model without dual-purchasing, by setting $\theta=0$ in our proposed model. 

\begin{theorem}\label{thm:MCPs-single}
Assume that the first part of Assumption~\ref{as:parameters} is satisfied (i.e., full market coverage), and that $\theta=0$ (i.e., users only single-purchase). Then, 
\begin{enumerate}
\item A $CP_j$, $j\in[2,M-1]$, with QoS/delay $d_j$ will have a total of $n_{j}=\frac{1}{M-1}+\frac{d_{j+1}+d_{j-1}-2d_j}{2t}+\frac{S_{j+1}+S_{j-1}-2S_j}{2t}$ users. 
\item A $CP_j$, $j\in\{1,M\}$, with QoS/delay $d_j$ will have a total of  $n_{j}=\frac{1}{2}(\frac{1}{M-1}+\frac{d_{j+(-1)^{\scriptsize{\mathbbm{1}\{j=M\}}}}-d_j}{t}+\frac{S_{j+(-1)^{\scriptsize{\mathbbm{1}\{j=M\}}}}-S_j}{t})$ users. 
\end{enumerate}
Further, let $n_j(d_j, d_{j\pm 1})$ be the number of users determined in (1) and (2) above, evaluated at a given QoS/delays profile $d$. Then, 
\begin{enumerate}
\setcounter{enumi}{2}
    \item The QoS offered to prioritized CPs' users will be determined by
    \begin{align*}
        d^*_j :=& \max_{d_{j}\in(0, d_0]} \sum_{j=1}^{M} \big(\frac{d_0-d_j}{t}(S_j+\lambda r_j) \notag\\
        & \qquad\qquad + \lambda n_j(d_j, d_{j\pm 1}) C(d_{j}))\big)~.
    \end{align*}
    \item $CP_j$ will be charged a prioritization fee of
    \begin{align*}
    p_j(d_j^*) = \frac{S_j+\lambda r_j}{\lambda t}\cdot \frac{d_0-d_j^*}{n_j(d_j^*, d^*_{j\pm 1})},
    \end{align*}
\end{enumerate}
\end{theorem}

\textbf{Interpretation:} From Parts (1) and (2), it is easy to see that prioritization of a $CP_j$'s competitor (decrease in $d_{j \pm 1}$) will \emph{reduce} the number of users of the CP. In contrast, when accounting for dual-purchasing, prioritization of competitors will only reduce the attention received by a CP, as this CP can still have a chance to serve as the secondary provider for his competitor's primary users. Further, from Parts (3) and (4), we observe that unlike Theorem~\ref{thm:MCPs}, prices charged to prioritized CPs depend on those charged to their competitors. In particular, from Part (4), we note that if $CP_j$'s competitor has paid for prioritization, then the ISP can charge a \emph{higher} prioritization fee to $CP_j$. This is because the ISP can exploit the fact that a CP is necessarily at a disadvantage (regardless of his business model) when competitors are prioritized, and hence has additional incentive/pressure to pay for prioritization. These conclusions are largely similar to prior work studying CP competition using single-purchasing Hotelling models~\cite{guo2017effects,cheng2011debate}. 

\textbf{Takeaway:} The negative impacts of prioritization of competitors on non-prioritized content is less severe when users purchase from multiple content providers. If users only single-purchased, non-prioritized content providers would lose users, and would have to pay higher prioritization fees for the same QoS once their competitors have been prioritized. 
}
\section{Discussion}\label{sec:discussion}

\subsection{Policy and Practical  Implications}\label{sec:practical-policy}

\paragraph{On the need for the ``no throttling'' assumption} We first note that an important assumption underlying our conclusions 
is that non-prioritized traffic is not throttled or otherwise slowed down by the ISP. From a policy perspective, this means that allowing paid prioritization, but imposing  no-throttling of the non-prioritized traffic, could lead to an outcome in which ISPs invest in additional infrastructure in a way that ultimately benefits users. A number of ISPs have in fact made implicit or explicit commitments to not throttle or block lawful content, even without net neutrality regulations~\cite{isp-promises}; this suggests that a ``no throttling'' regulation could be feasible from all stakeholders' perspective. 

We note that throttling of content (and also its prioritization) is only of importance if it is discernable by the users. 
While our current model considers a continuous QoS parameter $d_i$, our results will hold qualitatively with a discrete model of QoS sensitivity. 
In particular, non-prioritized CPs will not lose users if their users can not discern a drop in their contents' QoS. 

\paragraph{How should CPs pay for prioritization?} Our results suggest that CPs with sufficiently high revenue from either subscriptions or advertising can afford to pay prioritization fees. 
Although we have taken CPs' revenue rates to be fixed, this finding also implies that CPs can aim to increase their revenue (by increasing advertisements, improving their ad-display algorithms, increasing their subscription fees, or making their content more valuable to users) so as to be able to afford fast lanes. In particular, CPs could opt to offer both  ad-based and premium (ad-free) versions, with the latter providing subscribers with improved QoS. The paying subscribers would then be indirectly subsidizing the cost of additional infrastructure that supports their desired experience. However, the ability to introduce such fees, or increasing existing fees, will depend on the extent of competition and the subscription fees set by competitors in the content market. 

\paragraph{Who should be paying for additional Internet infrastructure?} In this paper, we have considered the construction of new infrastructure that is subsidized by CPs, in return for better QoS for their users. This approach is in line with views expressed by some ISPs 
who have suggested that large CPs should be paying the ISPs to build infrastructure that can support their traffic (by either cutting their profit margin, or charging \emph{their own} users sufficiently high subscription fees), as opposed to the ISP having to increase its access fees to \emph{all} users so as to be able to expand its infrastructure~\cite{attpolicy, frontier-free-internet}. The opposing view adopted by CPs 
is that the ISPs, and not the CPs, should be cutting their (already high) profit margins in response to the need for additional infrastructure. 
Regardless of which side pays for the infrastructure, the cost may still be passed down to users in form of increased fees (either ISPs' access fees or CPs' subscription fees). An analysis of changes in these fees, the effects of CP and ISP competition on their ability to adjust the fees, and the potential need for regulations to protect users from ultimately bearing the cost of added infrastructure, remain as interesting questions. 

\paragraph{On the trends of ISPs' spending in infrastructure} There has been substantial debate on how ISPs' investments in infrastructure are affected by regulations, with conflicting conclusions on whether broadband investments increased~\cite{forbes-down,r-street} or decreased~\cite{wired-not-true,sorry-ajit-drop} following the introduction of net neutrality laws, and also following the 2018 FCC ruling~\cite{vice-aftermath-investment}. 
We would also like to point out that while the institution of the type of paid prioritization arrangements we have considered in this paper could lead to additional infrastructure investments, the benefits will be limited to the prioritized CPs' users. Therefore, both short term and long term effects of the latest FCC rulings on improving broadband infrastructure and access for \emph{all users} remain to be seen. 
\subsection{Model Extensions and Future Directions}\label{sec:extensions}

We close this section by discussing potential extensions of our model, and opportunities for future work. 

We first elaborate on the potential effects of ISP competition on our conclusions. Due to the absence of ISP competition in our model, all infrastructure expansion costs are ultimately passed on to the CP. However, in a competitive ISP market, each ISP may also have an incentive to enter into exclusive prioritization agreements with (bigger) CPs, in a way that increases the ISP's users. 
As a result, ISPs may compete to attract CPs by partially (or fully) covering the costs of the required infrastructure investments. An analysis of such markets remains a main direction of our future work. 

We further note some potential alternatives for capturing CPs' competition. 
Our proposed model assumes that content providers are distributed along a \emph{line}. We adopted this choice as we believe it can capture some common forms of competition in the content market. For instance, in capturing ideological/political placement of news outlets, CPs' position could correspond to right, left, or centrist media. 
Alternatively, one can consider the distribution of CPs at equally distanced points on a circle's circumference (i.e. a Salop circle~\cite{anderson2019importance}), 
or more generally, let the $M$ providers sit at the vertices of an $M-1$-dimensional simplex, allowing users to fall at the intersection of any two (or more) types of content. 

In addition to altering the content providers' positions, our model can be modified to capture different user distributions {and demands}. In particular, similar to prior work, our analysis is based on a uniform distribution of users on the preference spectrum. Our model can be modified to place user mass away or closer to specific CPs, so as to capture perceived/observed user tendency towards certain content. 
\rev{We state the extension of Theorem~\ref{thm:MCPs} to general user distributions in Appendix~\ref{app:thm3-proof-inpaper}, and find that our results on changes in CPs' users and revenues continue to hold qualitatively.}
We leave further analysis with alternate user and content provider distributions as interesting directions of future work. 

\rev{Lastly, we note that our analysis has focused on competition between content providers having a non-zero mass of users prior to the introduction of paid prioritization options. This is reflective of the competition between established content providers. 
However, paid prioritization can disproportionately affect late entrants to content markets: as major CPs continually subsidize their own fast lanes, the default QoS available to non-prioritized content remains too low. In that case, an entrant CP (e.g. a startup) would have to offer a radically more valuable service (higher $V$) to be able to attract users. 
A model capturing the dynamics between existing vs future entrants remains as another interesting extension.} 
\section{Numerical Validation}\label{sec:simulations}
In this section, we first verify our results from Section~\ref{sec:analysis} through numerical simulations. We then discuss additional insights by considering the effects of different status quo QoS/delays $d_0$ on the ISP's revenue, and comparing the results of our model to those with CP-independent prioritization pricing schemes and prioritization through capacity re-allocation rather than expansion. 
We detail the choice of parameters used in these simulations in Table \ref{t:parameters}. 

\begin{table*}[t]
\caption{Choice of parameters for simulations. We consider the major streaming services (Netflix, Hulu, Disney+, etc), and choose our parameters based on information on their subscription fees, ad-revenues, user base, etc, as detailed below.}
\label{t:parameters}
\centering
\begin{tabular}{ | c | m{0.6\columnwidth} | m{0.9\columnwidth} |  c | } 
 \hline
 Notation & Description & Rationale for choice of value & Realistic value\\ 
 \hline
 $M$ & Number of CPs & Major competitors in streaming who are also content creators (potentially catering to different users' intrinsic preferences) are Netflix, Hulu, Amazon Prime, Apple TV+, Disney+, and HBO Max \cite{major-content-investing} & 6\\
 \hline
 $V$ & Users' base value from content consumption & We set this value relative to the access fees and subscription fees to capture users' voluntary participation in the market & 100\\
 \hline
 $t$ & Transport/fit cost, determining users' intrinsic preference for content & Studies  show that 46\% of U.S. broadband users subscribe to more than one over-the-top streaming service~\cite{mobile-ott}. We set $t$ to roughly get this level of multi-purchasing at the non-prioritized outcome & 500\\
 \hline
 $\theta$ & Users' residual benefit rate from secondary content & Based on a 2018 study saying Netflix users spend 10 hours/week and Amazon/Hulu users 5 hours/week on the platform~\cite{cnbc-hours} & 0.5\\ 
 \hline
 $F$ & Access fee charged to users by the ISP for connectivity & \$50 to \$60 monthly fees for broadband access (based on AT\&T and Comcast rates). Further, reports show that around 60\% of Internet traffic is from video streaming services~\cite{time-on-video}. We use this to estimate the portion of access fees related to streaming services 
 & 33\\ 
 \hline
 $S_j$ & Subscription fee charged to users by $CP_j$ & From \$6 to \$16, based on the range of Netflix and Hulu monthly subscription rates (Hulu's \$6 subscription includes ads) & 10 or [6, 16]\\
 \hline
 $d_0$ & Default QoS/delay for non-prioritized traffic & AT\&T's average broadband delay is 73ms~\cite{broadband-delay}, while the average 4G delay is around 54ms~\cite{cellular-delay}. To capture the effects of QoS on shaping users' choices, we scale this value to make it comparable against the base value and subscription fees  & 6 or [5.4, 7.3]\\
 \hline
 $r_j$ & $CP_j$'s revenue rate from displaying advertisements & We let the average value be around \$2.27, extrapolating from Hulu's \$1.5B annual revenue and 55M ad-supported viewers/month~\cite{digiday-ads} & 2.27 or [1, 4]\\ 
 \hline
 $\delta$ & Rate of reduction in ad-revenue from secondary users & We set this equal to $\theta$ based on \cite{cnbc-hours}, since both parameters are reflected in the time spent on the different CPs & 0.5 \\
 \hline
 $\lambda$ & Users' average traffic/demand rate & Netflix minimum speed is 0.5Mbps, 5Mbps for HD and 25 Mbps for Ultra HD. Our results are invariant to the scaling of this parameter & 5\\ 
 \hline
\end{tabular}
\end{table*}

As detailed in Table~\ref{t:parameters}, we choose our simulation parameters based on the competition between major streaming services who also invest substantially on content creation~\cite{major-content-investing}. By making different types of content available, these CPs are catering to different users' taste; our Hotelling model could hence be suitable for  modeling users' intrinsic preferences for the content offered by each service. Our choices of users' decision parameters in the Hotelling model and the CPs' revenue parameters are guided by information about these streaming services; see Table~\ref{t:parameters}. 

\begin{figure*}[t]
    \centering
\begin{minipage}[t]{0.48\textwidth}
\centering
   \includegraphics[width=0.72\columnwidth]{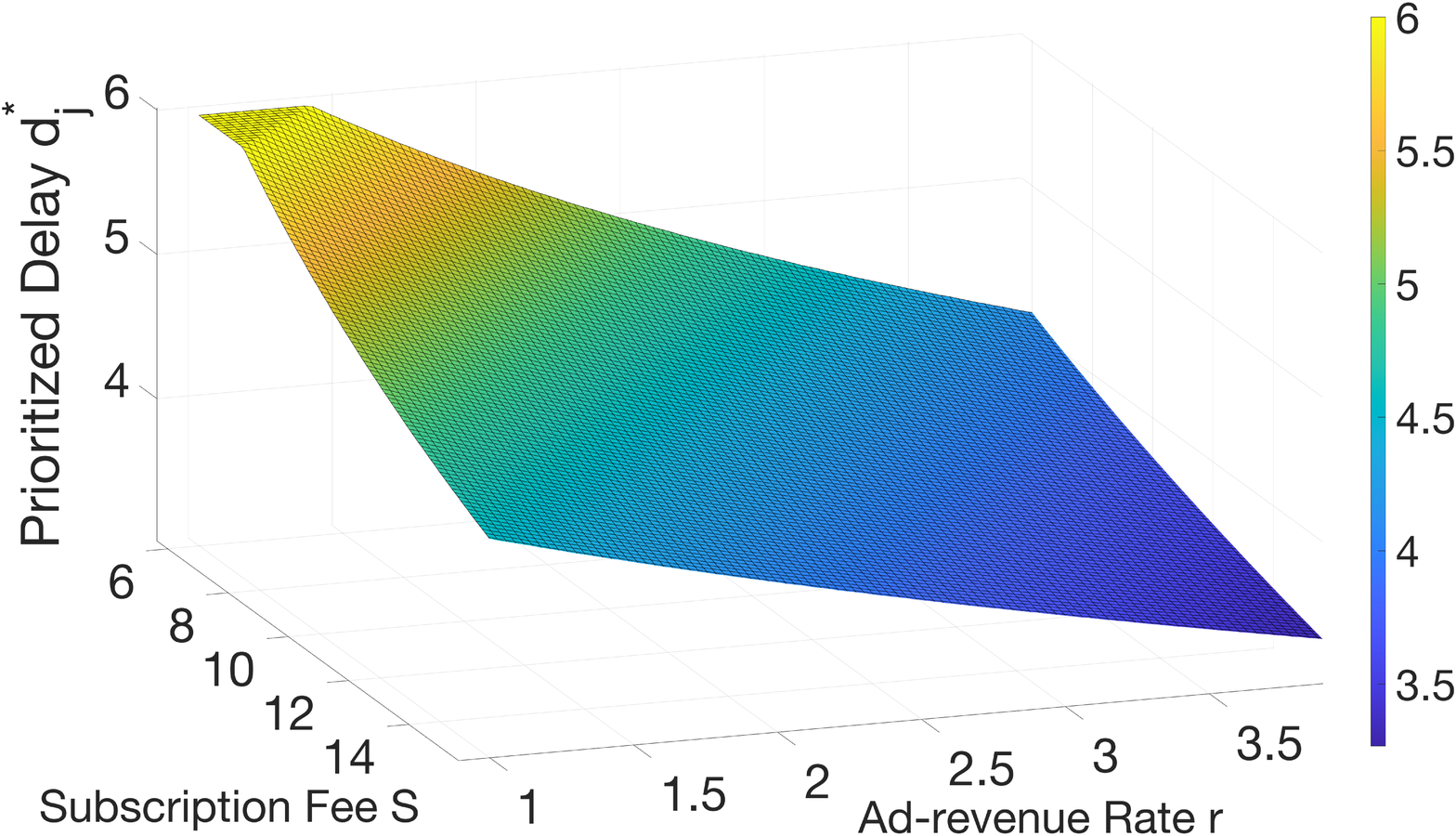}
   \caption{The prioritized CP's QoS/delay as a function of  $S_j$ and $r_j$, for $C(d)=1/d-1/d_0$. CPs with higher subscription fees or revenue rates can subsidize better QoS for their users.}
   \label{fig:delay-opt} 
\end{minipage}
\hspace{0.2in}%
\begin{minipage}[t]{0.48\textwidth}
\centering
   \includegraphics[width=0.72\columnwidth]{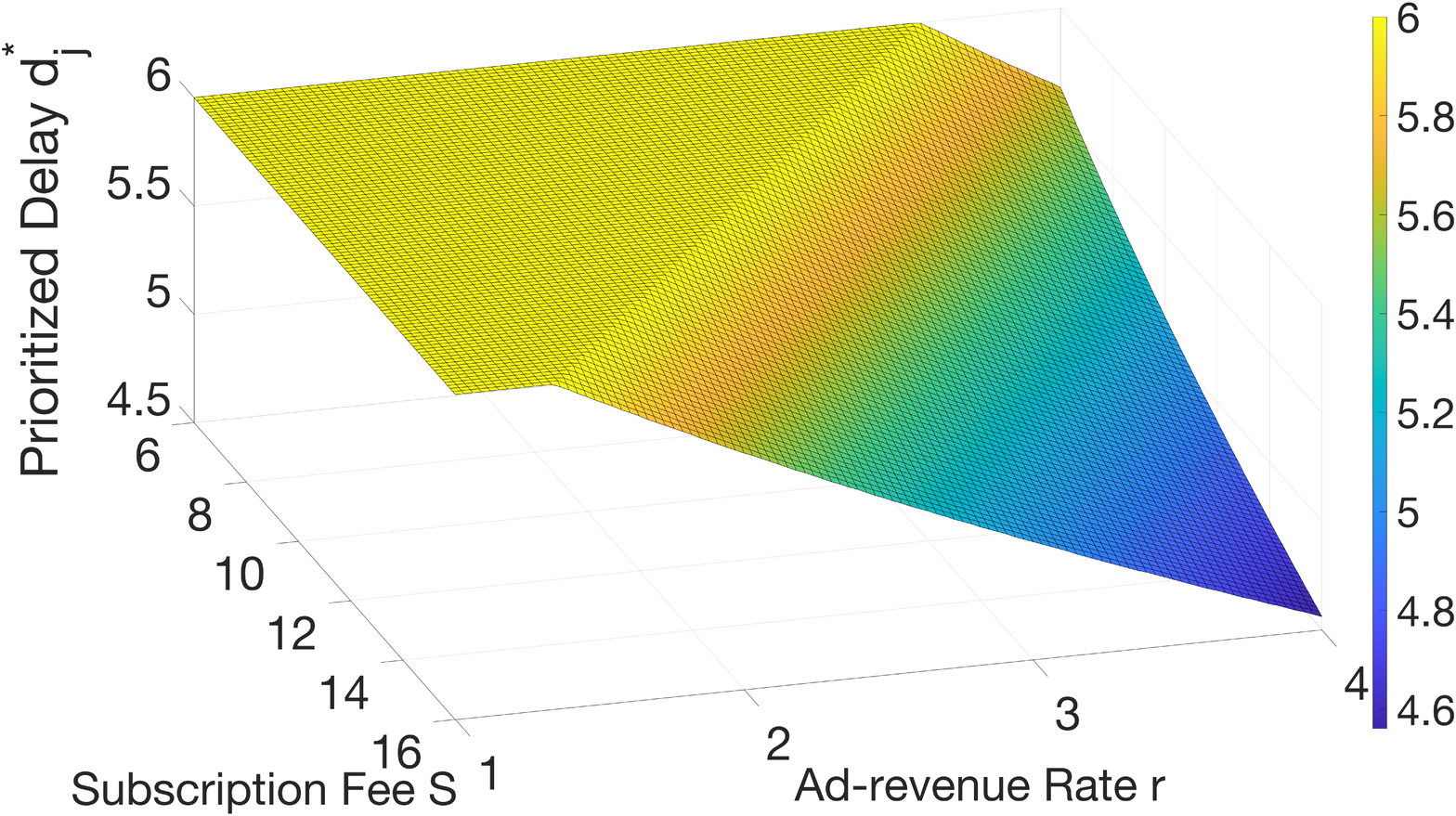}
   \caption{The prioritized CP's QoS/delay as a function of  $S_j$ and $r_j$, for $C(d)=2(1/d-1/d_0)$. Compared to Figure~\ref{fig:delay-opt}, CPs need higher revenue rates to afford prioritization, since it is more costly for the ISP to build new infrastructure.}
   \label{fig:delay-opt-2} 
\end{minipage}
\end{figure*}

\paragraph*{Offered QoS/delay $d^*$} We first consider the CPs' ability to afford fast lanes, and illustrate the QoS offered to users of fast lanes at the market equilibrium. We first consider the ISP cost function $C(d)=1/d - 1/d_0$. Figure~\ref{fig:delay-opt} illustrates the optimal $d_j^*$ offered to $CP_j$ as a function of $S_j$ and $r_j$, where the default QoS/delay is $d_0=6$, and we vary $S_j\in[6,16]$ and $r_j\in[1,4]$, guided by the estimates described in Table~\ref{t:parameters}. The numerical results are consistent with the solution to the ISP's optimization problem derived in part 3 of Theorem~\ref{thm:MCPs}. We observe that, as shown in Corollary~\ref{cor:affordability}, a $CP_j$ will only be able to afford fast lanes under sufficiently high $S_j$ or $r_j$, with the offered delay decreasing in both parameters. In particular, the flat top left of the surface indicates that for low subscription fees and ad-revnue rates, the CP can not afford prioritization and hence keeps the default QoS/delay. 

We conduct the same analysis for an ISP cost function $C(d)=2(1/d - 1/d_0)$, i.e., when it is more costly for the ISP to build fast lanes. The results are illustrated in Figure~\ref{fig:delay-opt-2}. As expected, the CPs need to have even higher subscription fees and revenue rates to afford prioritization.

\begin{figure*}[t]
    \centering
\begin{minipage}[t]{0.32\textwidth}
\centering
   \includegraphics[width=\columnwidth]{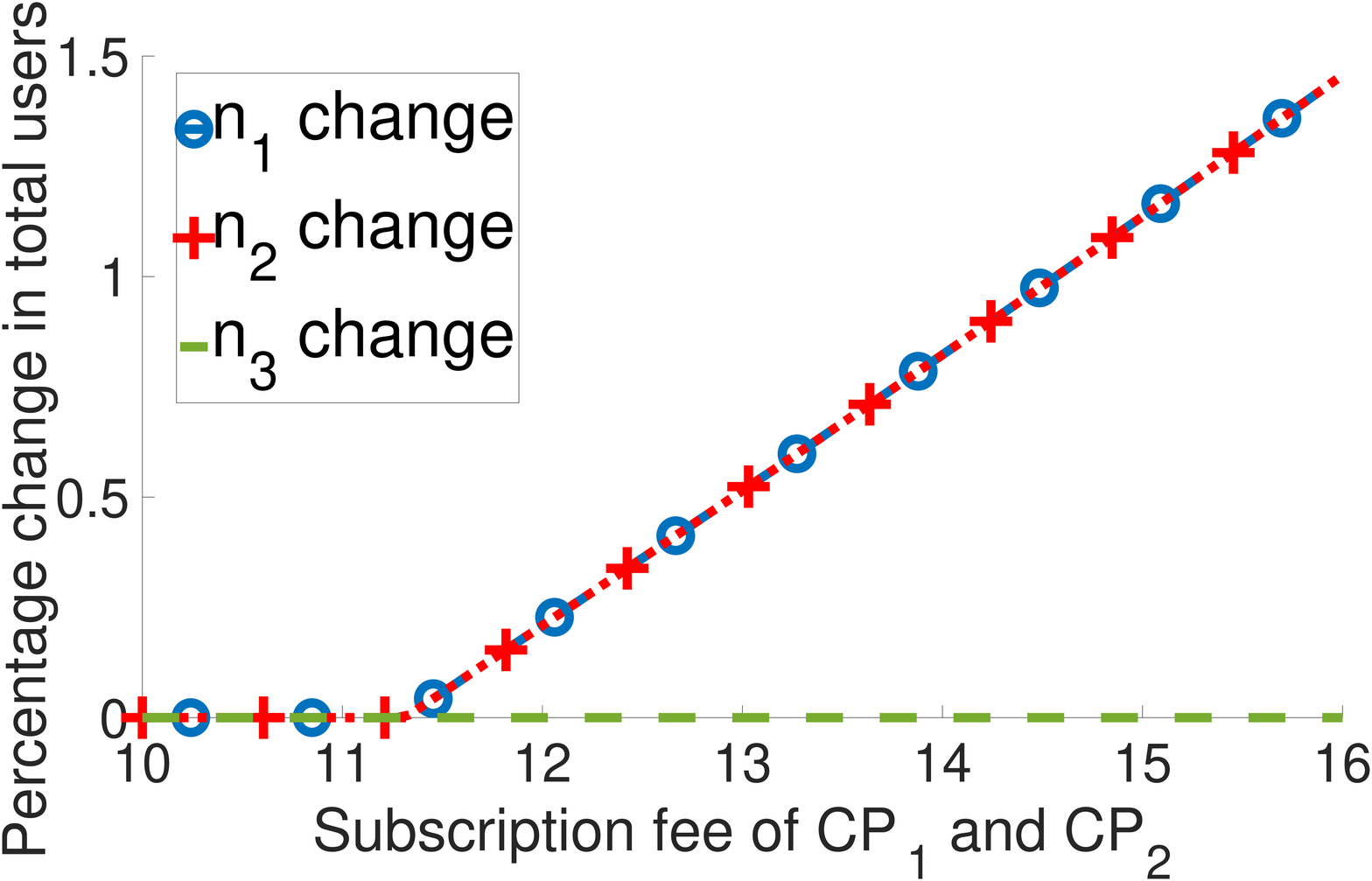}
   \caption{Changes in $CP_1, CP_2, CP_3$ total users following prioritization of $CP_1, CP_2$. Prioritized CPs gain additional users, while non-prioritized $CP_3$ can maintain the same total number of users.}
   \label{fig:cps-users}
\end{minipage}
\hspace{0.1in}%
\begin{minipage}[t]{0.32\textwidth}
\centering
   \includegraphics[width=\columnwidth]{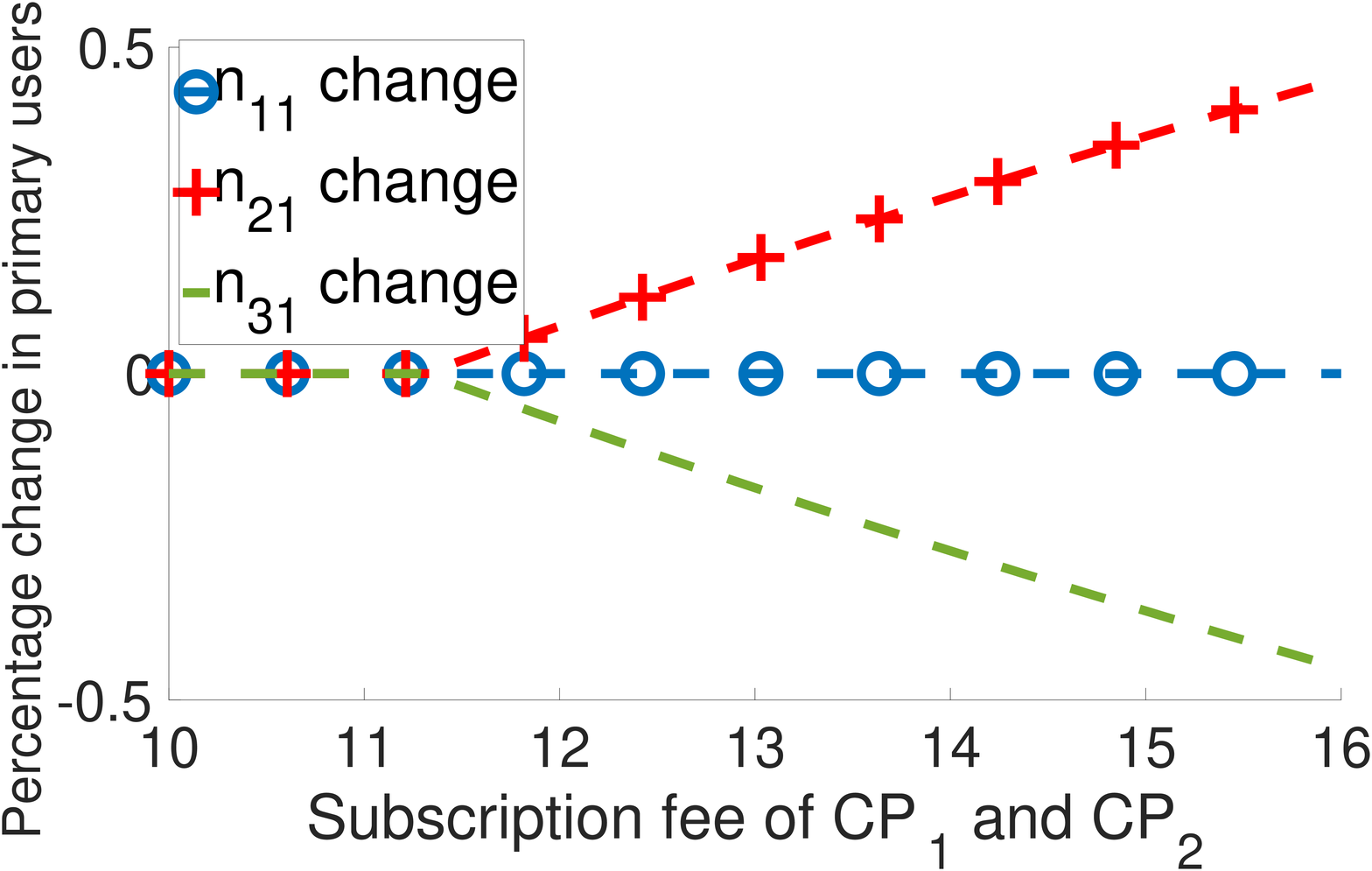}
   \caption{Changes in $CP_1, CP_2, CP_3$ primary users following prioritization of $CP_1, CP_2$. Non-prioritized $CP_3$ gets fewer primary users. Prioritized $CP_1$ will not gain additional attention as his competitor $CP_2$ is prioritized as well.}
   \label{fig:cps-users-prim}
\end{minipage}
\hspace{0.1in}%
\begin{minipage}[t]{0.32\textwidth}
\centering
   \includegraphics[width=\columnwidth]{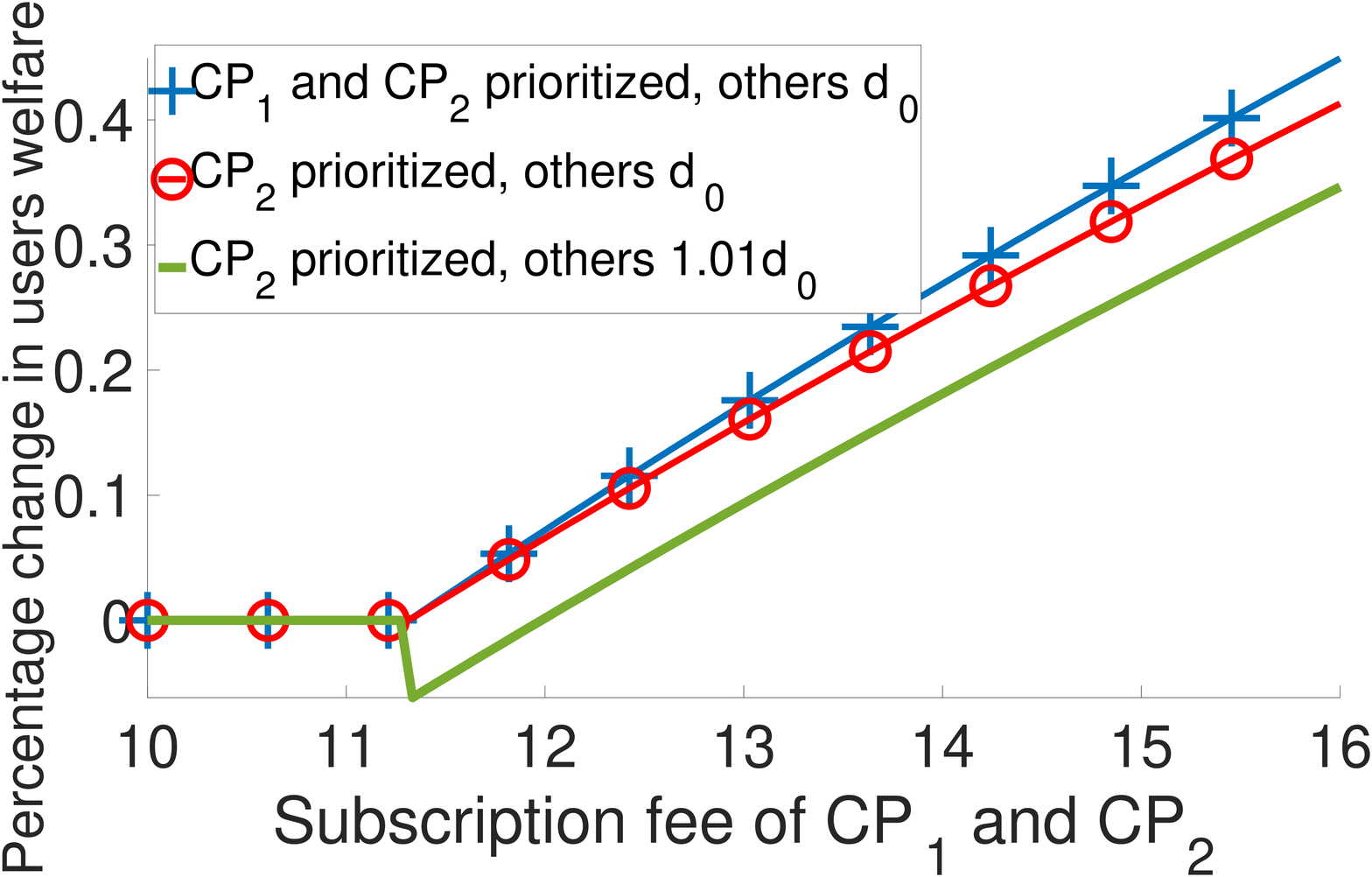}
   \caption{Change in users' welfare. Prioritization can increase users' welfare. However, if non-prioritized traffic is throttled, users' welfare may decrease.}
   \label{fig:welfare}
\end{minipage}
\end{figure*}

\paragraph*{Changes in CPs' users} We next illustrate the changes in CPs' users, in particular when only a subset of them have sufficiently high revenue rates to afford prioritized fast lanes. We consider the ISP cost function $C(d)=2(1/d - 1/d_0)$, and set $r=2.27$ for all CPs. We then vary $S\in[10,16]$ for $CP_1$ and $CP_2$, and fix $S=6$ for all other CPs. Then, the users of $CP_j, j\in\{3,\ldots, 6\}$ will have the default QoS/delay $d_0$, while those of $CP_1$ and $CP_2$ may be prioritized (c.f. Figure~\ref{fig:delay-opt-2}). Note also that $CP_1$'s (only) competitor  will be prioritized as well, while $CP_2$ will still have a non-prioritized competitor, $CP_3$. 

Figures~\ref{fig:cps-users} and~\ref{fig:cps-users-prim} illustrate the changes in $CP_1$, $CP_2$, and $CP_3$ total users and primary users, respectively, as $CP_1$ and $CP_2$'s increasing subscriptions enable them to afford prioritization, while $CP_3$ remains non-prioritized. 
As we would expect from Corollary~\ref{cor:cp-users-change}, \rev{despite users' larger subscription fees $S_1$ and $S_2$,} $CP_1$ and $CP_2$ gain users as they pay for greater prioritization, while the non-prioritized $CP_3$ retains the same total number of users (Figure~\ref{fig:cps-users}), yet loses primary users (Figure~\ref{fig:cps-users-prim}). Note also that, consistent with Remark~\ref{r:non-exclusive}, $CP_1$ does not gain any primary users as he has \emph{the same} prioritization contract as his competitor $CP_2$. However, $CP_2$ can gain primary users due to his improved standing compared to 
$CP_3$.


\paragraph{Users' welfare} Next, we consider users' welfare under prioritization. We consider the same parameters as those used to compare the changes in CPs' users above, so that only $CP_1$ and $CP_2$ may afford prioritization. 
In Figure~\ref{fig:welfare}, we illustrate the change in users' welfare for three scenarios: prioritization of both $CP_1$ and $CP_2$, exclusive prioritization of $CP_2$, and exclusive prioritization of $CP_2$ while throttling the non-prioritized contents by 1\%, so that $d_j=1.01d_0, j\neq 2$. These observations are consistent with  Corollary~\ref{cor:user-welfare}. In particular, users' welfare increases with prioritization, with higher welfare if more CPs are prioritized. 
However, if non-prioritized traffic is slowed down, as shown in Figure~\ref{fig:welfare}, prioritization of one content while throttling competing ones may result in a reduction of users' welfare. Note also that for sufficiently high $S_2$, more users adopt the prioritized content and receive considerably better QoS compared to the default setting, hence increasing overall welfare in spite of throttling of non-prioritized content.

\rev{
\paragraph{Effects of the initial QoS/delay $d_0$ on the ISP's revenue} We next consider the effects of the status quo of the service offered in the non-prioritized regime on the change in the ISP's revenue once paid prioritization is offered. Figure~\ref{fig:isp-revenue} illustrates the increase in the  ISP's revenue relative to the non-prioritized regime, as a function of $d_0$, for $S_j\in\{6,10,16\}$, when the ISP cost is given by $C(d)=1/d-1/d_0$. We observe that the ISP can benefit from prioritization when the initially offered QoS is lower. In particular, if $d_0$ is lower to begin with, the ISP's cost for further improving QoS is too high, and hence no prioritization contracts are offered. Further, we see that the ISP's revenue is increasing in the CPs'  subscription fees. This trend, which is consistent with Corollary~\ref{cor:cp-revenue}, is because the ISP can charge higher prioritization prices to higher revenue CPs.
}

\begin{figure*}[t]
    \centering
\begin{minipage}[t]{0.32\textwidth}
\centering
   \includegraphics[width=\columnwidth]{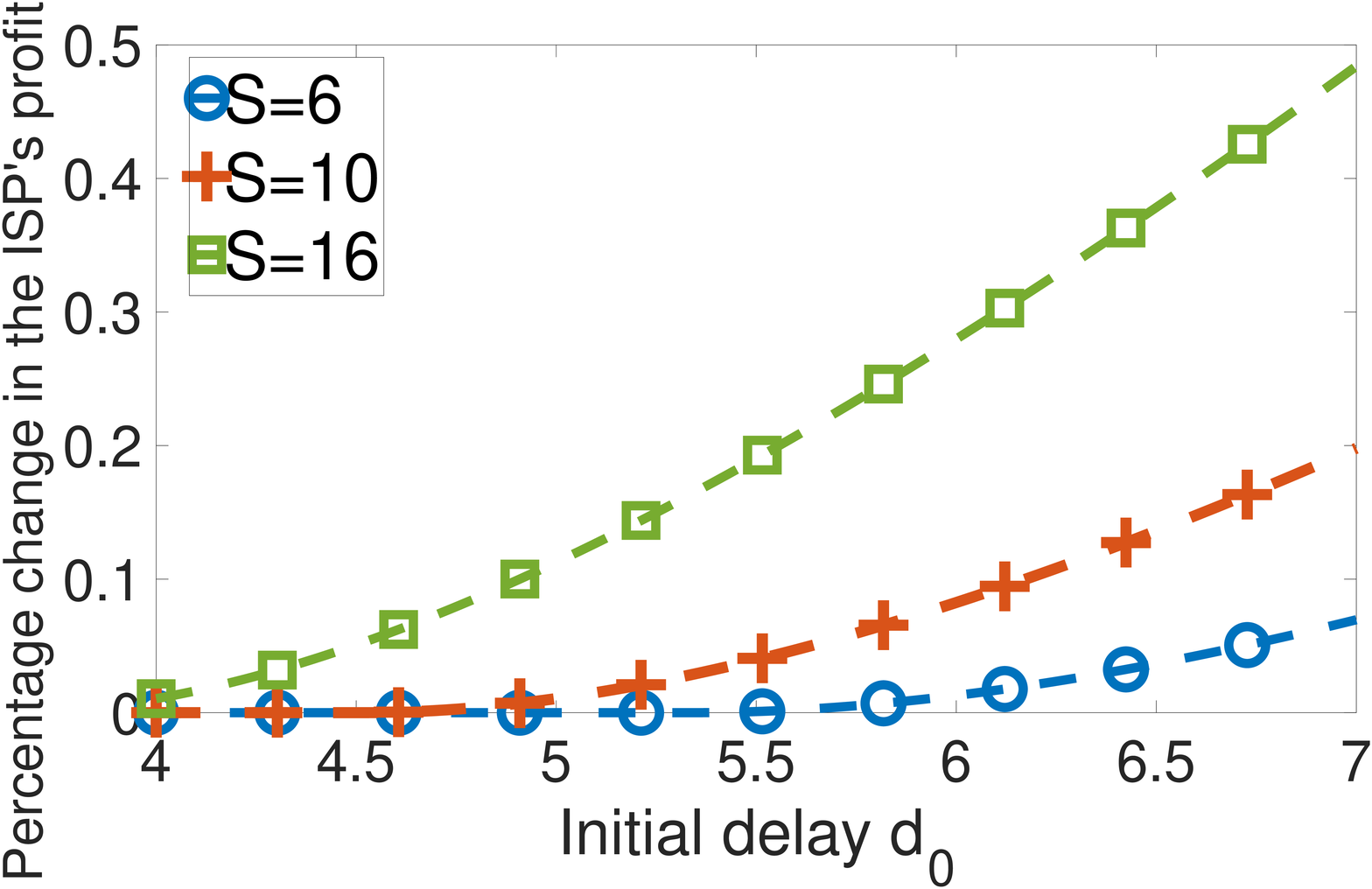}
   \caption{Change in the ISP's revenue as a function of initial QoS/delay $d_0$. The ISP benefits more from offering prioritization when the default QoS is lower (higher delay $d_0$), with the profit increasing in CPs' subscription fees.}
   \label{fig:isp-revenue}
\end{minipage}
\hspace{0.1in}%
\begin{minipage}[t]{0.32\textwidth}
\centering
   \includegraphics[width=\columnwidth]{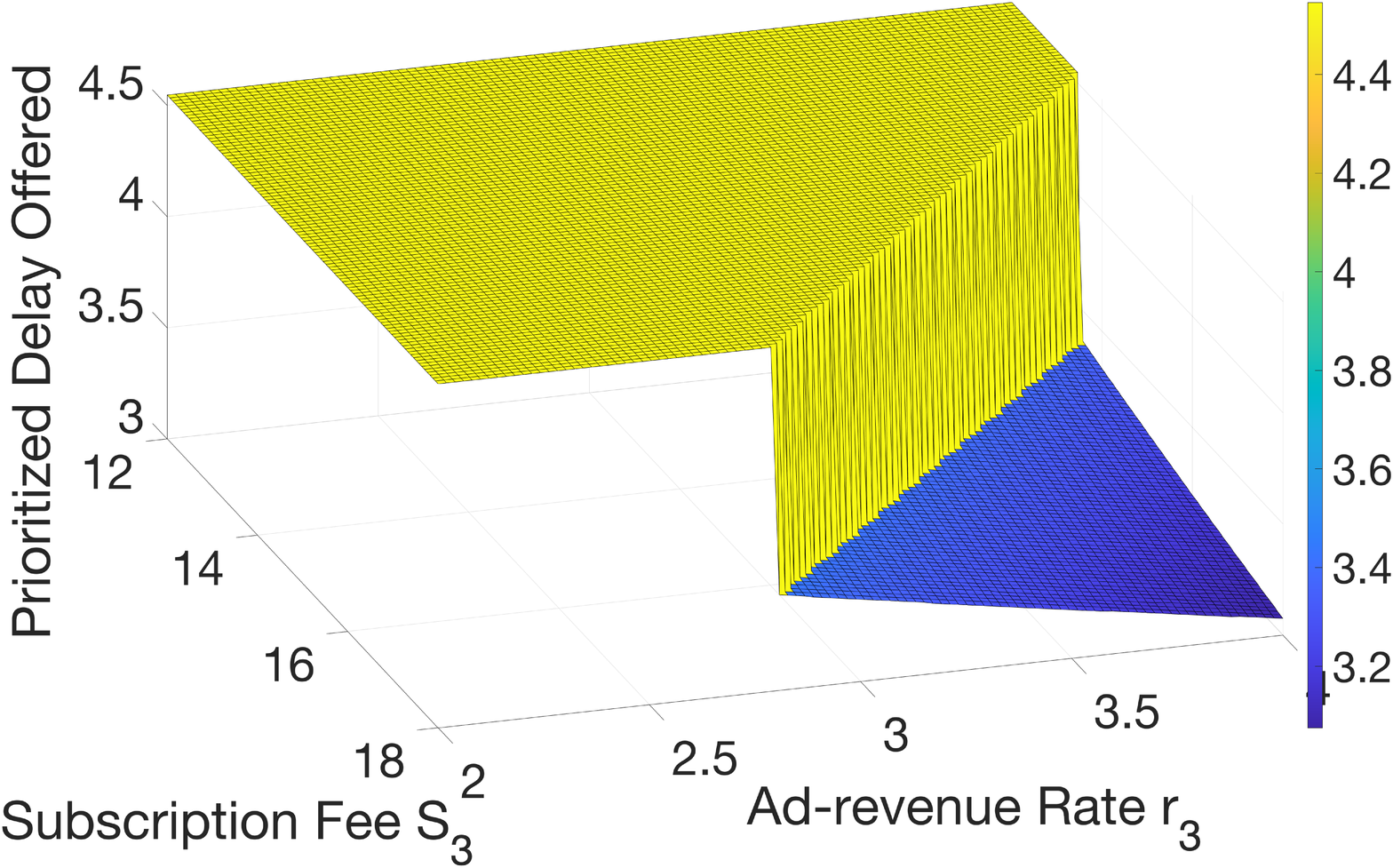}
   \caption{Changes in equilibrium prioritization delay when prioritization pricing is CP-independent. Once $CP_3$ has a sufficiently higher revenue rate than other CPs, the ISP focuses her offerings on meeting the demand of this dominant CP's users.}
   \label{fig:uniformprice}
\end{minipage}
\hspace{0.1in}%
\begin{minipage}[t]{0.32\textwidth}
\centering
   \includegraphics[width=\columnwidth]{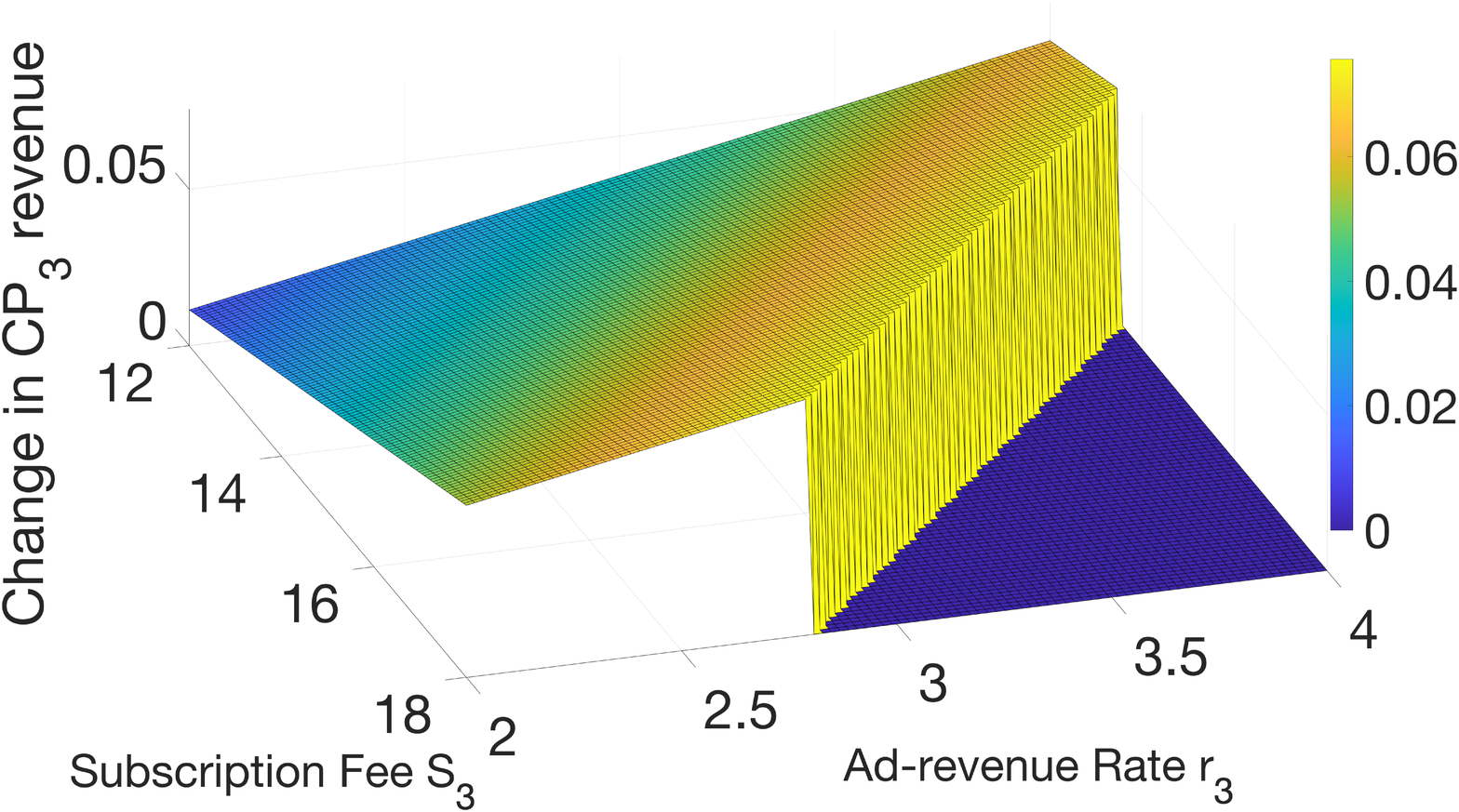}
   \caption{Changes in the revenue of a dominant $CP_3$ when 
   prioritization pricing is CP-independent. Although uniform pricing results in lower QoS than attainable by $CP_3$, this CP can derive higher profit due to lower prices set to attract other CPs.}
   \label{fig:uniformprice-2}
\end{minipage}
\end{figure*}

\rev{
\paragraph*{Offering CP-independent  prioritization options} So far, we have assumed that the ISP can tailor the pricing charged for prioritization contracts to each CP, so as to extract the maximum profit from each provider. We now relax this assumption, by imposing that the ISP should offer a single menu of prioritization prices $p(d)$, without discriminating between content providers. In particular, we let $CP_3$ be a content provider who has higher subscription fees and/or ad-revenue rates compared to other CPs, and hence can offset the ISP's cost for construction of higher quality fast lanes. In this case, the ISP can decide whether to offer a high priced menu to extract the maximum profit from this CP, or to offer a lower priced menu to get more CPs to opt for prioritization. 

Figures~\ref{fig:uniformprice} and~\ref{fig:uniformprice-2} illustrate the offered fast lane QoS/delays, and the revenue of $CP_3$, respectively, for $C(d)=1/d-1/d_0$, $S_3\in[12,18]$, and $r_3\in[2,4]$. We observe that once $CP_3$ has  significantly higher revenue rates compared to other CPs, the ISP switches her offering of prioritization, so as to extract higher profit by (only) meeting the demands of this CP's users (Figure~\ref{fig:uniformprice}). Further, as shown in Figure~\ref{fig:uniformprice-2}, even though uniform pricing results in lower QoS than attainable by $CP_3$, this CP can derive higher profit in the meantime as a result of lower prices set to meet other CPs' demands.

These results suggest that if one of the streaming providers is significantly more profitable than other competitors, and the ISP can not tailor her offerings to different CPs, then she will target her prioritization pricing to profit from this CP alone. Interestingly, if prioritization offerings could be tailored to CPs, then all other CPs could negotiate  prioritization contracts as well (albeit with lower quality fast lanes). Together with our findings in Corollary~\ref{cor:user-welfare} and Figure~\ref{fig:welfare}, this means that offering tailored prioritization options will \emph{improve} the welfare of users. 
}

\begin{figure*}[t]
    \centering
\begin{minipage}[t]{0.32\textwidth}
\centering
   \includegraphics[width=\columnwidth]{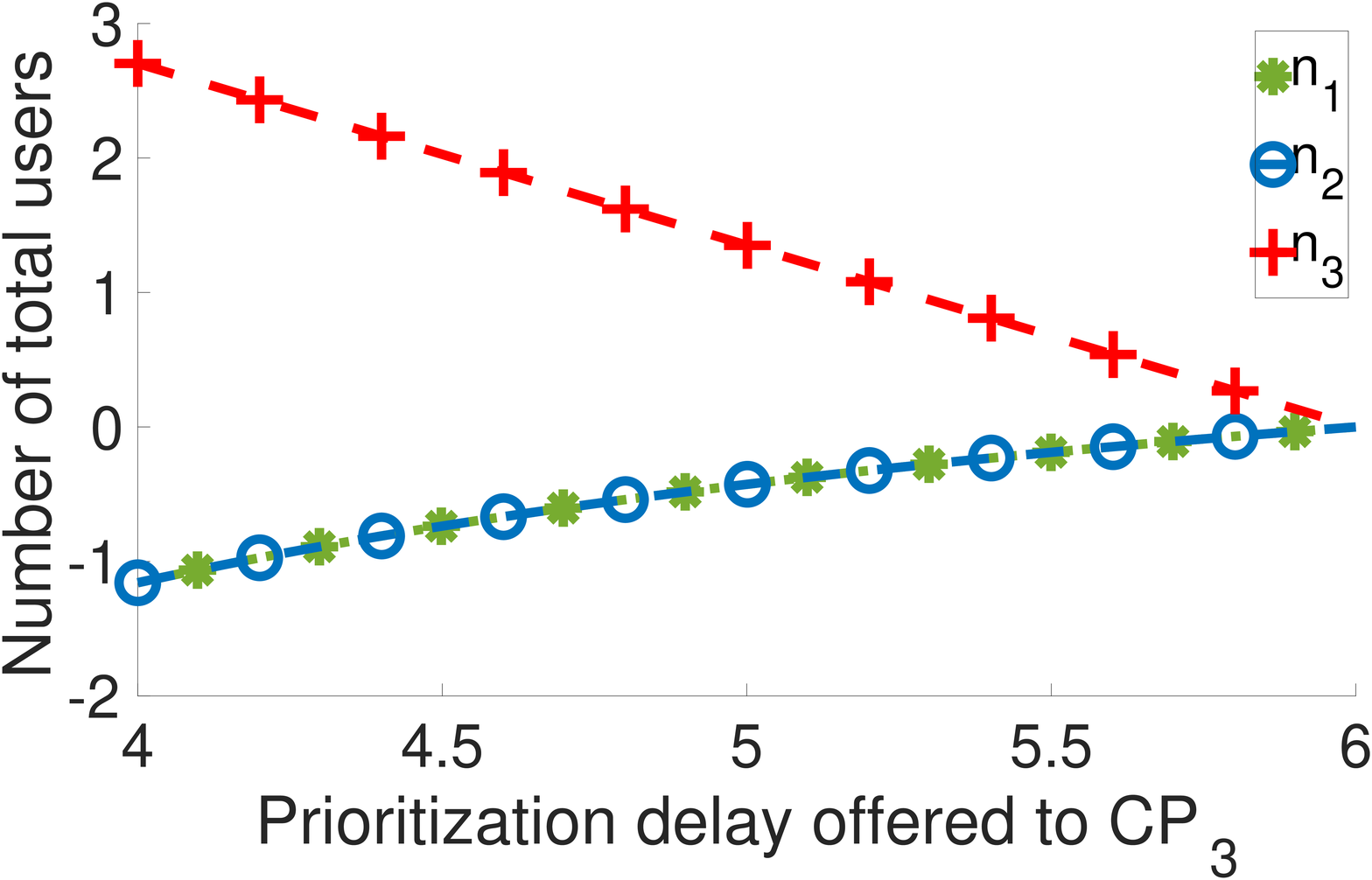}
   \caption{Change in CPs' total number of users, when $CP_3$ is prioritized through capacity reallocation. Non-prioritized content gets lower users as it is inevitably throttled.}
   \label{fig:allocation-users}
\end{minipage}
\hspace{0.1in}%
\begin{minipage}[t]{0.32\textwidth}
\centering
   \includegraphics[width=\columnwidth]{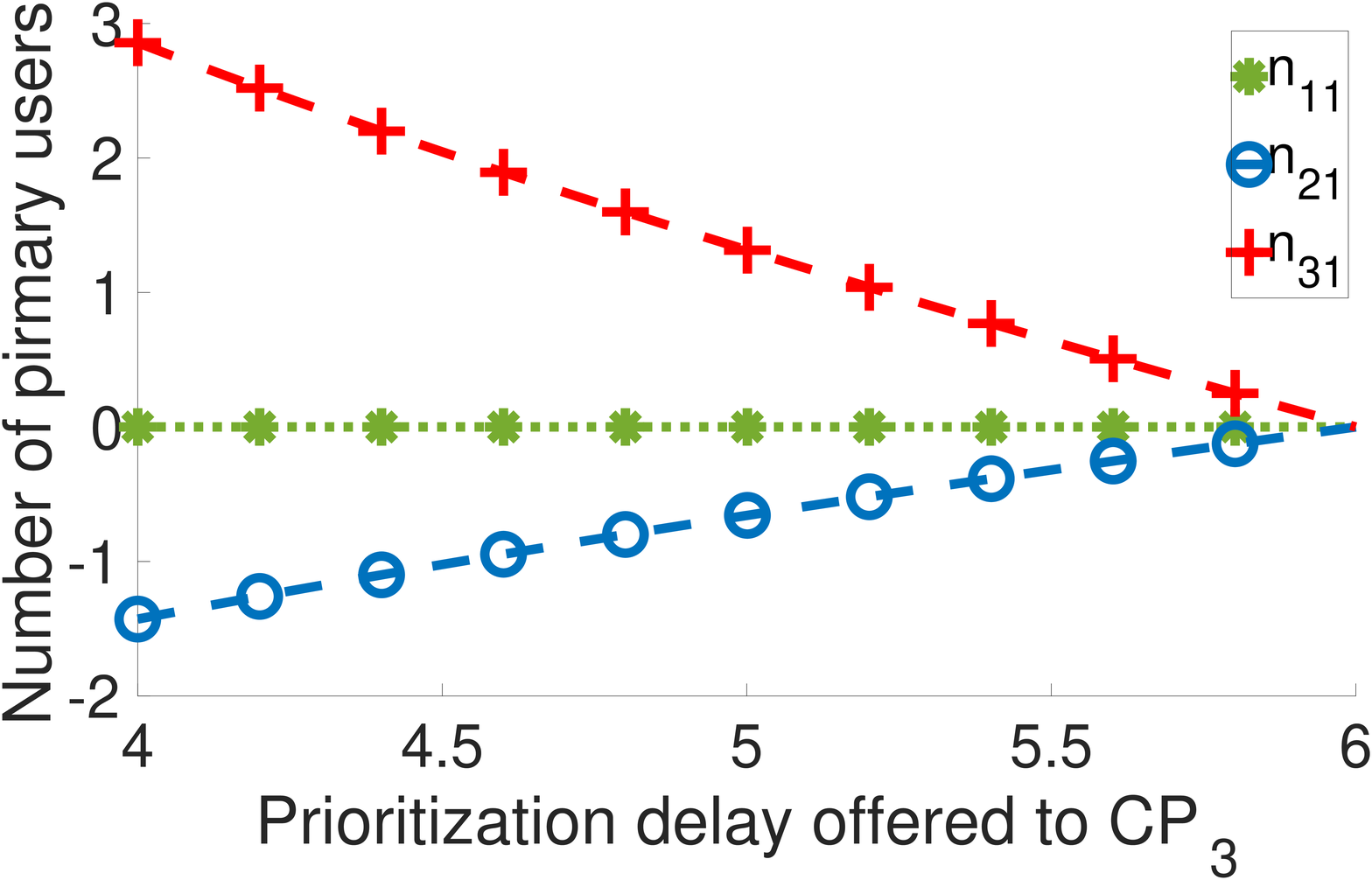}
   \caption{Change in CPs' primary users, when $CP_3$ is prioritized through capacity reallocation. Non-prioritized $CP_2$, the direct competitor of $CP_3$, is losing primary users.}
   \label{fig:allocation-primary}
\end{minipage}
\hspace{0.1in}%
\begin{minipage}[t]{0.32\textwidth}
   \centering
   \includegraphics[width=\columnwidth]{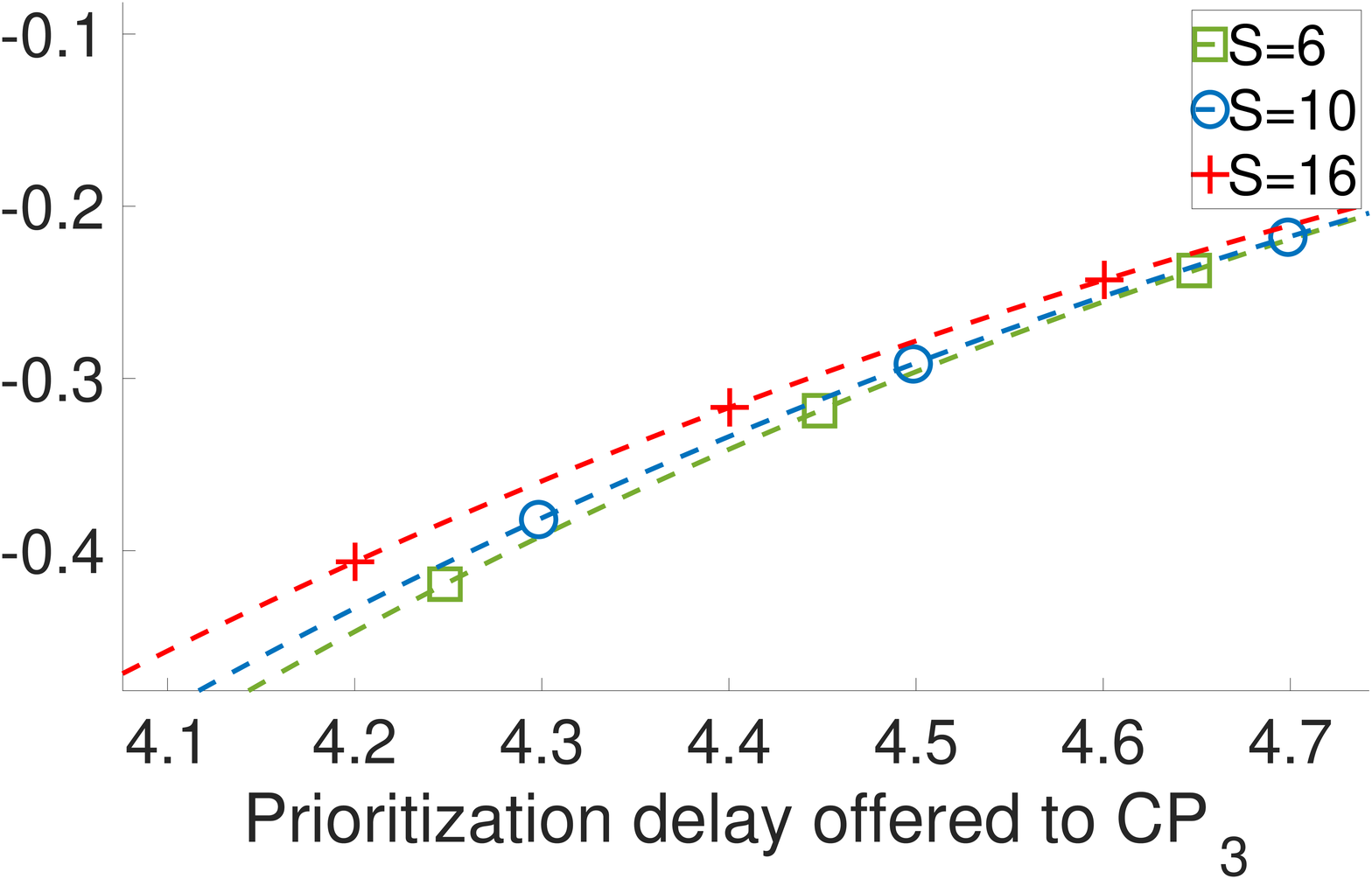}
   \caption{Welfare decreases as one of the content providers, $CP_3$, is prioritized through capacity reallocation, which inevitably leads to throttling of non-prioritized traffic. }
   \label{fig:allocation-welfare}
\end{minipage}
\end{figure*}

\rev{
\paragraph{Prioritization through capacity re-allocation.} Finally, we allow the ISP to offer improved QoS by re-allocating her existing capacity, rather than through building additional infrastructure. This will inevitably result in throttling of non-prioritized traffic. In particular, we consider reservation-based (PMP~\cite{odlyzko1999paris}) delays of the form $d_j={n_j}/{\Phi_j}$, where $\Phi_j$ is the network capacity dedicated to $CP_j$, and $n_j$ is the number of $CP_j$ users. We assume one of the middle CPs, say $CP_3$, is offered an exclusive prioritization contract with improved QoS/delay $d_3\in[4,6]$. To provide this fast lane, the ISP dedicates additional capacity $\Phi_3$ to this CP, at the expense of reduced capacity (hence, worse QoS) for other CPs' users. 

Figures~\ref{fig:allocation-users} and~\ref{fig:allocation-primary}  
show the change in the total number of users and primary users, respectively, of $CP_1$, $CP_2$, and $CP_3$ as a function of the delay offered to $CP_3$ who is exclusively prioritized. First, from Figure~\ref{fig:allocation-users}, we observe that non-prioritized CPs attain fewer users at equilibrium due to the degradation of their QoS. 
Figure~\ref{fig:allocation-primary} shows that $CP_1$ will not lose primary users (as it is competing with $CP_2$ who has the same degraded QoS), yet $CP_2$ loses primary users to his higher QoS competitor $CP_3$. 
As expected, the gap increases as $CP_3$ is offered a better prioritization deal (i.e., lower $d_3$). 

Figure~\ref{fig:allocation-welfare} illustrates the change in users' welfare following exclusive prioritization of $CP_3$ for three different values of $S_3$. We observe that users' welfare decreases when $CP_3$ is prioritized due to the throttling, and hence reduced QoS experienced, by other CPs' users. However, if $S_3$ has high subscription fees, he can attract fewer users (even if prioritized), and therefore the exclusive prioritization of his content will be less disruptive to other CPs' users. 
}
\section{Conclusion}\label{sec:conclusion}
We studied the effects of paid prioritization on an ISP's incentives for infrastructure investments in a market with content provider competition. We proposed a new model that captures users' heterogeneous preference over competing CPs, as well as their  multi-purchasing of content. We show that 
there indeed exist incentives for larger CPs (with high subscription and/or advertisement revenue) to offset the costs of fast lane construction. More importantly, our proposed model allowed us to analyzed the effects of prioritization on competing CPs with different revenue models, showing that the resulting CP revenue is just as dependent on the source of revenue as the size of the CP. 
In particular, we show that non-prioritized CPs need not lose users, but may have reduced ad-revenue due to decreased attention from their users. We establish that such arrangements can increase users' welfare as long as non-prioritized traffic is not throttled, and discuss the policy and practical implications of our findings. 

\rev{Our multi-purchasing model, which enabled us to account for users' attention to different content, and hence differentiate between the competing CPs' subscription and advertisement revenue sources, could provide new insights in other studies of paid prioritization and net neutrality implications. These include extending the current model to also capture ISP competition, 
and analyzing CPs' incentives for modifying their business models.} 

\bibliographystyle{IEEEtran}
\bibliography{sdp-prioritization}

\newpage
\appendix
\rev{

\subsection{Users' choice of content consumption}\label{app:content-choice-users}

We begin by identifying users' choices of content consumption (without imposing any assumptions on the problem parameters). We will use this lemma to derive the conditions of Assumption \ref{as:parameters} and in proving Theorem \ref{thm:MCPs}. We use the following definition throughout the proof. 
\begin{definition}\label{def:thresholds-app}
Define the primary and secondary thresholds for $CP_j$ as 
\begin{align*}
    \tau^P_j := \frac{1}{t}(V_j - d_j - S_j)~,\\
    \tau^S_j := \frac{1}{t}(V_j - d_j - \frac{S_j}{\theta})~,
\end{align*}
respectively. 
\end{definition}
Note that both these thresholds are functions of $d_j$, and that $\tau^P_j>\tau^S_j$. These thresholds denote the maximum distance from $CP_j$ (as defined in the Hotelling model) at which a user would derive utility from consumption of content from the CP as its primary and secondary content choice, respectively. We use these indifference points to determine users' content consumption choices between any two CPs, as shown in the following lemma. Note that users who are single/dual purchasing between two content providers $CP_j$ and $CP_k$ have four consumption options: $\{j\},\{j,k\},\{k,j\},\{k\}$,  where $\{j,k\}$ (resp. $\{j\}$) is a bundle in which the user first (resp. only) consumes content from $CP_j$, and vice versa. Note also that the following lemma only identifies the best choice among these four options; users will still choose to opt out of purchasing their best bundle unless it provides them with non-negative utility.

\begin{lemma}\label{lemma:content-choices-users}
Consider two content providers $CP_j$ and $CP_k$, $k>j$. Denote the CPs' positions in the Hotelling model by $X_j:=\frac{j-1}{M-1}$ and $X_k:=\frac{k-1}{M-1}$. Then, the choices of users $X_j<x<X_k$ between the four potential consumption options is as follows:
\begin{enumerate}

    \item If $X_j+\tau^S_j<X_j+\tau^P_j<X_k-\tau^P_k<X_k-\tau^S_k$:\\
    all users with $x<X_j+\tau^P_j$ single purchase on $CP_j$ and those with $x>X_k-\tau^P_k$ single purchase on $CP_k$.
    
    \item If $X_j+\tau^S_j<X_k-\tau^P_k<X_j+\tau^P_j<X_k-\tau^S_k$:\\
    all users with $x<\frac12 (X_k-\tau^P_k) + \frac12 (X_j+\tau^P_j)$ single purchase on $CP_j$ and those with $x>\frac12 (X_k-\tau^P_k) + \frac12 (X_j+\tau^P_j)$ single purchase on $CP_k$.
    
    \item If $X_k-\tau^P_k<X_j+\tau^S_j<X_j+\tau^P_j<X_k-\tau^S_k$, there are two possible cases:
    \begin{enumerate}
        \item If $\frac12 (X_k-\tau^P_k) + \frac12 (X_j+\tau^P_j)<X_j+\tau_j^S$, then users with $x<\frac12 (X_k-\tau^P_k) + \frac12 (X_j+\tau^P_j)$ single purchase on $CP_j$, those with $\frac12 (X_k-\tau^P_k) + \frac12 (X_j+\tau^P_j)<x<X_j+\tau_j^S$ dual purchase with primary $CP_{k}$, and those with $X_j+\tau_j^S<x$ single purchase on $CP_k$. 
        \item If $X_j+\tau_j^S<\frac12 (X_k-\tau^P_k) + \frac12 (X_j+\tau^P_j)$, then users with $x<\frac12 (X_k-\tau^P_k) + \frac12 (X_j+\tau^P_j)$ single purchase on $CP_j$, those with $\frac12 (X_k-\tau^P_k) + \frac12 (X_j+\tau^P_j)<x$ single purchase on $CP_k$. 
    \end{enumerate}
    
    \item If $X_k-\tau^P_k<X_j+\tau^S_j<X_k-\tau^S_k<X_j+\tau^P_j$, there are three possible cases:
    \begin{enumerate}
        \item If $\frac12 (X_k-\tau^P_k) + \frac12 (X_j+\tau^P_j)<X_j+\tau_j^S$, then users with $x<\frac12 (X_k-\tau^P_k) + \frac12 (X_j+\tau^P_j)$ single purchase on $CP_j$, those with $\frac12 (X_k-\tau^P_k) + \frac12 (X_j+\tau^P_j)<x<X_j+\tau_j^S$ dual purchase with primary $CP_{k}$, and those with $X_j+\tau_j^S<x$ single purchase on $CP_k$. 
        \item If $X_j+\tau_j^S<\frac12 (X_k-\tau^P_k) + \frac12 (X_j+\tau^P_j)<X_k-\tau_k^S$, then users with $x<\frac12 (X_k-\tau^P_k) + \frac12 (X_j+\tau^P_j)$ single purchase on $CP_j$, those with $\frac12 (X_k-\tau^P_k) + \frac12 (X_j+\tau^P_j)<x$ single purchase on $CP_k$. 
        \item If $X_k-\tau_k^S<\frac12 (X_k-\tau^P_k) + \frac12 (X_j+\tau^P_j)$, then users with $x<X_k-\tau_k^S$ single purchase on $CP_j$, those with $X_k-\tau_k^S<x<\frac12 (X_k-\tau^P_k) + \frac12 (X_j+\tau^P_j)$ dual purchase with primary $CP_j$, and those with $\frac12 (X_k-\tau^P_k) + \frac12 (X_j+\tau^P_j)<x$ single purchase on $CP_k$.
    \end{enumerate}
    
    \item If $X_k-\tau^P_k<X_k-\tau^S_k<X_j+\tau^S_j<X_j+\tau^P_j$, there are three cases:
    \begin{enumerate}
        \item If $\frac12 (X_k-\tau^P_k) + \frac12 (X_j+\tau^P_j)<X_k-\tau_k^S$, then users with $x<\frac12 (X_k-\tau^P_k) + \frac12 (X_j+\tau^P_j)$ single purchase on $CP_j$, those with $\frac12 (X_k-\tau^P_k) + \frac12 (X_j+\tau^P_j)<x<\frac12(X_j+X_k)-\frac{(V_k-V_j)-(d_k-d_j)}{2t}$ dual purchase with primary $CP_j$, those with $ \frac12(X_j+X_k)-\frac{(V_k-V_j)-(d_k-d_j)}{2t}<x<X_j+\tau_j^S$ dual purchase with primary $CP_{j+1}$, and those with $X_j+\tau_j^S<x$ single purchase on $CP_k$.
        \item If $X_k-\tau_k^S<\frac12 (X_k-\tau^P_k) + \frac12 (X_j+\tau^P_j)<X_j+\tau_j^S$, then users with $x<X_k-\tau_k^S$ single purchase on $CP_j$, those with $X_k-\tau_k^S<x<\frac12(X_j+X_k)-\frac{(V_k-V_j)-(d_k-d_j)}{2t}$ dual purchase with primary $CP_j$, those with $\frac12(X_j+X_k)-\frac{(V_k-V_j)-(d_k-d_j)}{2t}<x<X_j+\tau_j^S$ dual purchase with primary $CP_k$, and those with $x>X_j+\tau_j^S$ single purchase on $CP_k$. 
        \item If $X_j+\tau_j^S<\frac12 (X_k-\tau^P_k) + \frac12 (X_j+\tau^P_j)$, then users with $x<X_k-\tau_k^S$ single purchase on $CP_j$, those with $X_k-\tau_k^S<x<\frac12(X_j+X_k)-\frac{(V_k-V_j)-(d_k-d_j)}{2t}$ dual purchase with primary $CP_j$, those with $ \frac12(X_j+X_k)-\frac{(V_k-V_j)-(d_k-d_j)}{2t}<x<\frac12 (X_k-\tau^P_k) + \frac12 (X_j+\tau^P_j)$ dual purchase with primary $CP_{j+1}$, and those with  $\frac12 (X_k-\tau^P_k) + \frac12 (X_j+\tau^P_j)<x$ single purchase on $CP_{j+1}$. 
    \end{enumerate}        
\end{enumerate}
\end{lemma}
\begin{proof}
The proof involves comparing the utility that users derive at each consumption bundle, and finding the one that maximizes their payoff. The utility from each of the four bundles is given by:

\begin{align*}
U_x(\{j\}) & =  V_j - tm(x,j) - d_{j} - S_j - F~,\\
U_x(\{j,k\}) & =  V_j - tm(x,j) - d_{j} - S_j + \notag\\
&\hspace{0.2in} \theta(V_k - tm(x,k) - d_{k}) - S_{k} - F~,\notag\\
U_x(\{k,j\}) & =  V_k - tm(x,k) - d_{k} - S_{k} + \notag\\
&\hspace{0.2in} \theta(V_j - tm(x,j) - d_{j}) - S_j - F~,\\
U_x(\{k\}) & =  V_k - tm(x,k) - d_{k} - S_{k} - F~.
\end{align*}
A comparison across these leads to the lemma. 
\end{proof}
}
\subsection{Derivation of the conditions in Assumption \ref{as:parameters}}\label{app:assumption1}

We are primarily interested in analyzing a market in which both single-purchasing and dual-purchasing options can attract a non-zero mass of users, so as to analyze the effects of prioritization on changes in all possible consumption bundles. We will further be interested in outcomes under which the users have (on average) a higher number of content options to choose from, implying the highest level of competition between the CPs, so that users are not trivially defaulting to certain bundles due to lack of other options. 

From Lemma \ref{lemma:content-choices-users} we observe that for such outcomes, the parameters should be consistent with those under part (5). Further, the combination in (5-b) provides us with the set of parameters under which the users have on average a higher number of bundles to choose from, implying the highest level of competition between the CPs. We therefore derive the conditions under which the user choices under (5-b) in Lemma \ref{lemma:content-choices-users} are realized. We repeat the user choices here for each of reference: for users with $X_j<x<X_{j+1}$, with $V_j=V_{j+1}$: 
\begin{itemize}
    \item Users with $X_j<x<X_{j+1}-\tau_{j+1}^S$ single purchase on $CP_j$, 
    \item Users with $X_{j+1}-\tau_{j+1}^S<x<\frac12(X_j+X_{j+1})+\frac{d_{j+1}-d_j}{2t}$ dual purchase with primary $CP_j$,
    \item Users with $\frac12(X_j+X_{j+1})+\frac{d_{j+1}-d_j}{2t}<x<X_j+\tau_j^S$ dual purchase with primary $CP_{j+1}$,
    \item Users with $X_j+\tau_j^S<x<X_{j+1}$ single purchase on $CP_{j+1}$. 
\end{itemize}

\emph{Ensuring both single and dual purchasing:}
First, note the two extremes at which users single purchase. In order to ensure non-zero mass on these types of users, we need to have $\tau_j^S<\frac{1}{M-1}, \forall j$. This reduces to
\[(2-UB)\quad \theta < \frac{\min_j S_j}{V-\frac{t}{M-1}}\]

In addition, we require  $X_{j+1}-\tau_{j+1}^S<\frac12(X_j+X_{j+1})+\frac{d_{j+1}-d_j}{2t}<X_j+\tau_j^S$ for the two dual purchasing regions to be well-defined. Moreover, we require that $X_{j+1}-\tau_{j+1}^S<\frac12 (X_{j+1}-\tau^P_{j+1}) + \frac12 (X_j+\tau^P_j)<X_j+\tau_j^S$ to be within region (5-b) to begin with. These two conditions together simplify to requiring that
\begin{align*}
    X_{j+1} - X_j &< \min\{2\tau_j^S - \frac{d_{j+1}-d_j}{t} - \frac{S_{j+1}-S_j}{t}, \\
    & 2\tau_{j+1}^S + \frac{d_{j+1}-d_j}{t} + \frac{S_{j+1}-S_j}{t},\\
    & 2\tau_{j}^S - \frac{d_{j+1}-d_j}{t},  2\tau_{j+1}^S + \frac{d_{j+1}-d_j}{t}\}
\end{align*}
We minimize the upper bound above in order to ensure the condition is satisfied for all pairs of CPs. This reduces to finding $\theta$ such that 
\[\frac{t}{M-1}<2V-2d_0-\max_j S_j-2\frac{\max_j S_j}{\theta}\]
Therefore, the lower bound on $\theta$ is given by
\[(2-LB)\quad \theta > \frac{\max_j S_j}{V-d_0-\frac12(\max_j S_j +\frac{t}{M-1})}~.\]

\emph{High base value from content:} Lastly, we derive the full user market coverage condition and CP competition conditions. Specifically, we want all users, including the ones with the lowest utility, to derive positive benefit from participation (i.e., full market coverage), and not be forced to choose a CP when doing so for maximal competition (i.e., have at least two CP choices). This can be expressed as:
\begin{align*}
    \min_{\{x,j\}}\{V- tm(x,j)-d_{j} -S_j - F\}\geq 0
\end{align*}
In finding the minimum user utility, we note that this would require users to be at maximum distance from one of their two CP choices closest to them,  $m(x,j)=\frac{1}{M-1}$, have the worst case QoS/delay, $d_0$, and pay the highest CP subscription fee, $S_j$. This leads to requiring
\[(1)\quad V \geq \frac{t}{M-1} + d_0 + \max_j S_j + F\]
In terms of the notation of Lemma ~\ref{lemma:content-choices-users}, this condition is equivalent to requiring $\tau^P_j\geq \frac{1}{M-1}$ for all $CP_j$. 

\subsection{Proof of Theorem \ref{thm:MCPs}}\label{app:thm1-proof-inpaper}

\paragraph*{Users' content consumption decisions}
We begin by finding users' consumption patterns. Take users with $\frac{j-1}{M-1}\leq x\leq \frac{j}{M-1}$ (that is, those who are primarily interested in $CP_j$ and $CP_{j+1}$'s content). These users will have four potential consumption patterns: $\{j\}$, $\{j+1\}$, $\{j,j+1\}$ and $\{j+1,j\}$, where $\{j,-j\}$ (resp. $\{j\}$) is a bundle in which the user first (resp. only) consumes content from $CP_j$. 

Let $\tau^S_{j}(d_j) = \frac{1}{t}(V-d_j-\frac{S_j}{\theta})$. Under Assumption \ref{as:parameters}, the consumption choices of these users can, as shown in Appendix~\ref{app:content-choice-users}, is given by
\begin{itemize}
    \item Users with $\frac{j-1}{M-1}<x<\frac{j}{M-1}-\tau^S_{j+1}(d_{j+1})$ single purchase on $CP_j$, 
    \item Users with $\frac{j}{M-1}-\tau^S_{j+1}(d_{j+1})<x<\frac{j-\frac12}{M-1}+\frac{d_{j+1}-d_j}{2t}$ dual purchase with primary $CP_j$,
    \item Users with $\frac{j-\frac12}{M-1}+\frac{d_{j+1}-d_j}{2t}<x<\frac{j-1}{M-1}+\tau^S_j(d_j)$ dual purchase with primary $CP_{j+1}$,
    \item Users with $\frac{j-1}{M-1}+\tau^S_j(d_j)<x<\frac{j}{M-1}$ single purchase on $CP_{j+1}$. 
\end{itemize}

Therefore, depending on the relative delays $d_{j}$ and $d_{j+1}$ once prioritization decisions are made, users' content consumption can be illustrated as in Figure \ref{fig:users-consumption}. 

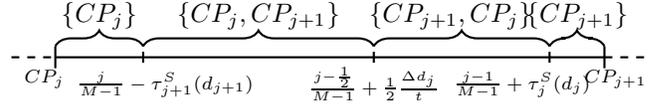
\begin{figure}[h!]
\begin{center}
\begin{tikzpicture}[scale=0.73]
\node[align=center] at (-0.2,-0.4)
{\scriptsize{$CP_j$}};
\node[align=center] at (10.2,-0.4)
{\scriptsize{$CP_{j+1}$}};
\node[align=center] at (5.8,-0.5) {\scriptsize{$\frac{j-\frac12}{M-1}+\frac12\frac{\Delta d_{j}}{t}$}};
\node[align=center] at (2,-0.5) {\scriptsize{$\frac{j}{M-1}-\tau^S_{j+1}(d_{j+1})$}};
\node[align=center] at (8.5,-0.5)
{\scriptsize{$\frac{j-1}{M-1}+\tau^S_{j}(d_j)$}};

\draw [thick,-] (0,0) -- (10,0);
\draw [thick,-] (0,0.2) -- (0,-0.2);
\draw [thick,-] (10,0.2) -- (10,-0.2);
\draw [thick,-] (5.8,0.1) -- (5.8,-0.1);
\draw [thick,-] (1.6,0.1) -- (1.6,-0.1);
\draw [thick,-] (9,0.1) -- (9,-0.1);

\draw [thick,dashed] (-0.8,0) -- (0,0);
\draw [thick,dashed] (10,0) -- (10.8,0);

\draw [thick,decorate,decoration={brace,amplitude=6pt,raise=0pt}] (0,0.2) -- (1.6,0.2);
\node[align=center] at (0.8,0.75) {$\{CP_j\}$};

\draw [thick,decorate,decoration={brace,amplitude=6pt,raise=0pt}] (1.6,0.2) -- (5.8,0.2);
\node[align=center] at (3.7,0.75) {$\{CP_j,CP_{j+1}\}$};

\draw [thick,decorate,decoration={brace,amplitude=6pt,raise=0pt}] (5.8,0.2) -- (9,0.2);
\node[align=center] at (7.2,0.75) {$\{CP_{j+1},CP_j\}$};

\draw [thick,decorate,decoration={brace,amplitude=6pt,raise=0pt}] (9,0.2) -- (10,0.2);
\node[align=center] at (9.5,0.75) {$\{CP_{j+1}\}$};
\end{tikzpicture}
\end{center}
    \caption{Users' content consumption choices.}
    \label{fig:users-consumption}
\end{figure}

\paragraph*{CPs' prioritization decisions}
Given users' consumption choices, we observe that if a $CP_j$ lowers his offered delay $d_{j}$, then he will serve a larger mass of the user market. This includes changes in both primary and secondary users. 
We evaluate the changes in a CP's revenue when opting for prioritization (accounting for the change in his users), so as to determine the incentive of the CP for paying prioritization fees. 

We first consider the revenue each $CP_j$ can attain from the users he is attracting under competition with one of his neighbors, $CP_{j+1}$. For any given pair of delays $d_j$ and $d_{j+1}$, from the $\frac{1}{M-1}$ users that reside between $CP_j$ and $CP_{j+1}$ in the Hotelling model, $CP_j$ serves a total of  $ \tau^S_{j}(d_j) = \frac{1}{t}(V-d_j-\frac{S_j}{\theta})$ users, with $\frac12(\frac{1}{M-1}+\frac{d_{j+1}-d_{j}}{t})$ of them being primary users, and the remainder  as secondary users.
The portion of $CP_j$ revenue that is accrued through this competition is therefore given by 
\begin{align}
\bar{R}_{j,j+1}(d_j)&=\tau^S_j(d_j)(S_j + \lambda\delta r_j - \lambda z_jp_j(d_j))\notag\\ 
&\qquad +\lambda r_j(1-\delta)\frac12(\frac{1}{M-1}+\frac{d_{j+1}-d_j}{t})~.
\label{eq:Rj-Mcp-j+1}
\end{align}
Note that the revenue when $CP_j$ does not opt for prioritization can be obtained by setting $d_j=d_0$ and $z_j=0$ in \eqref{eq:Rj-Mcp-j+1}. Similarly, $CP_j$'s portion of revenue accrued from competition with $CP_{j-1}$, $\bar{R}_{j,j-1}(d_j)$, can be derived analogously by swapping the indices $j+1$ with $j-1$ in \eqref{eq:Rj-Mcp-j+1}. Therefore, the total revenue to a middle-CP is given by 
\begin{align}
    R_j(d_j) &= 2\tau^S_j(d_j)(S_j + \lambda\delta r_j - \lambda z_jp_j(d_j))\notag\\ 
&~ +\lambda r_j(1-\delta)(\frac{1}{M-1}+\frac{d_{j+1}+d_{j-1}-2d_j}{2t})~. \label{eq:Rj-Mcp}
\end{align}
For the two end-CPs, $CP_1$ and $CP_M$, who only have one competitor, the total revenue is given by $R_1(d_1)=\bar{R}_{1,2}(d_1)$ and $R_M(d_M)=\bar{R}_{M,M-1}(d_M)$, respectively. 

Then, $CP_j$ will have an incentive to pay the required prioritization price if the CP's revenue given in \eqref{eq:Rj-Mcp}
increases following prioritization compared to his no prioritization revenue at $d_0$. In particular, the prioritization price for delay $d_j$ should be such that $R_j(d_j)\geq R_j(d_0)$. Simplifying this, the price charged to $CP_j$ for delay $d_j$ should satisfy the following constraint: 
\begin{align}
p_j(d_{j}) &\leq \frac{(\tau^S_{j}(d_j)-\tau^S_j(d_0))(S_j+\lambda\delta r_j)+ \frac{1}{2}(1-\delta)\lambda r_j\frac{d_0-d_j}{t}}{\lambda\tau^S_{j}(d_j)}\notag\\
& = \frac{S_j+\frac{1}{2}(1+\delta)\lambda r_j}{\lambda t}\cdot \frac{d_0-d_j}{\tau^S_{j}(d_j)}~.
\label{eq:cp-ir}
\end{align}

\paragraph*{The ISP's choices}
Lastly, the ISP's profit is given by
\begin{align*}
    \Pi &= F + \sum_{j=2}^{M-1} \lambda z_{j} 2\tau^S_{j}(d_j)(p_j(d_j) - C(d_{j})) \notag\\
    &\qquad + \sum_{j\in\{1,M\}} \lambda z_{j} \tau^S_{j}(d_j)(p_j(d_j) - C(d_{j}))~.
\end{align*}
where $C(d_j)$ is the per-user, per-unit of traffic, cost for building additional infrastructure to provide a delay $d_{j}$ to the $2\tau^S_{j}(d_j)$ users (or $\tau^S_{j}(d_j)$ users for end CPs) of a prioritized provider $CP_j$. 

We first note that as the ISP's profit is increasing in the prioritization price of any CP, she will charge the maximum possible price the CP is willing to pay. That is, the price charged to $CP_j$ will bind \eqref{eq:cp-ir}. 

The ISP further needs to determine her proposed QoS/delay $d_j$ to $CP_j$. Substituting for the price binding \eqref{eq:cp-ir}, this will be given by the solution to
\begin{align*}
    \max_{d_{j}\in(0, d_0]} ~ \frac{d_0-d_j}{t}(S_j+\frac12\lambda(1+\delta) r_j) - \lambda \tau^S_{j}(d_j) C(d_{j})~.
\end{align*}
If the prioritization arrangements are exclusive, the ISP will make the offer to the CP(s) from whom she can extract the most profit. 


\subsection{Proof of Corollary \ref{cor:cp-users-change}}
Assume $CP_j$ has opted for prioritization to get QoS/delay $d_j^*$, while all other CPs maintain the based QoS/delay $d_0$. From Theorem \ref{thm:MCPs} and Lemma \ref{lemma:cp-users}, the change in $CP_j$'s users is given by
\[2\tau_j(d^*_j)-2\tau_j(d_0)=2\frac{d_0-d^*_j}{t}~.\]
The change in primary users is given by
\[n_{1j}(d^*_j) - n_{1j}(d_0) = \frac{d_{j+1}+d_{j-1}-d_{j}^*}{2t}~.\]
As long as the other CPs are not prioritized, the above indicates a change of $\frac{d_0-d_j^*}{t}$ in $CP_j$'s primary users. Therefore, half the total gains are in primary users, and the remainder are in secondary users. A similar application of the theorem establishes the result for end-CPs. 

\subsection{Proof of Corollary \ref{cor:cp-revenue}}
For part (1), from the proof of Theorem \ref{thm:MCPs}, we know that the ISP's revenue is strictly increasing in the price charged to CPs for prioritization. The CPs' incentive for paying for prioritization on the other hand is given by $R_j(d_j)\geq R_j(d_0)$. Therefore, the ISP charges each CP for improved delays in a way that that his revenue at any offered delay is equal to $R_j(d_0)$. 

For part (2), from the proof of Theorem \ref{thm:MCPs}, the revenue of $CP_{j+1}$ is given by
\begin{align}
    R_{j+1}(d_{j+1}) &= \notag\\ &2\tau^S_{j+1}(d_{j+1})(S_{j+1} + \lambda\delta r_{j+1} - \lambda z_{j+1}p_{j+1}(d_{j+1}))\notag\\ 
&~ +\lambda r_{j+1}(1-\delta)(\frac{1}{M-1}+\frac{d_{j+2}+d_{j}-2d_{j+1}}{2t})~. \label{eq:Rj+1-cor2}
\end{align}
We assume $d_j=d_j^*$, and that $d_k=d_0, \forall k\neq j$. By inspecting  \eqref{eq:Rj+1-cor2}, we note that the change in the revenue of $CP_{j+1}$ from changes of $d_j$ stems from the last term. When $CP_j$ is prioritized, the change in $CP_{j+1}$'s revenue will be given by
\[\lambda r_{j+1}(1-\delta)\frac{d_0-d^*_{j}}{2t}\]
A similar argument holds for $CP_{j-1}$.

\subsection{Proof of Corollary \ref{cor:affordability}}
The ISP's choice of delay is determined from 
\begin{align*}
    d^*_j =&   \arg\min_{d_{j}\in(0, d_0]} ~
         \frac{S_j+\frac{1}{2}(1+\delta)\lambda r_j}{t}d_j + \lambda\tau_j(d_j)C(d_j).
\end{align*}
The first derivative of this (convex) objective function is given by
\[\frac{S_j+\frac{1}{2}(1+\delta)\lambda r_j}{t} + \lambda\tau_j(d_j)C'(d_j) - \lambda \frac{1}{t} C(d_j)~.\]
For the optimization problem to have a root $d_j^*>d_0$, we need for the first derivative to be positive at $d_0$, so that it is optimal to decrease the delay. That is
\[\frac{S_j+\frac{1}{2}(1+\delta)\lambda r_j}{t} + \lambda\tau_j(d_0)C'(d_0) > 0~.\]
Or equivalently, that
\[|C'(d_0)|<\frac{S_j+\frac{1}{2}(1+\delta)\lambda r_j}{t\lambda\tau_j(d_0)}~.\]

In addition, the optimal QoS/delay $d_j^*$ satisfies
\[(S_j+\frac{1}{2}(1+\delta)\lambda r_j) + \lambda(V-d^*_j-\frac{S_j}{\theta})C'(d^*_j) - \lambda  C(d^*_j)=0~.\]
As both $C(d)$ and $C'(d)$ are decreasing in $d$, increasing the constant $r_j$ decreases the root $d_j^*$ of the above. Further, noting that $C'<0$, increasing $S_j$ also increases the LHS of the equation above, and therefore the root $d_j^*$ of the above will decrease. Both of these indicate that the CP with higher $S_j$ or $r_j$ can afford better QoS/lower delay. 

\subsection{Proof of Corollary \ref{cor:user-welfare}}
Users' welfare can be found as the summation of the welfare of the primary and secondary users of each CP. For instance, for users with $\frac{j-1}{M-1}\leq x\leq \frac{j}{M-1}$ (those between $CP_j$ and $CP_{j+1}$), the welfare of the primary users of $CP_j$ is
\[\int_{0}^{\tiny{\frac{0.5}{M-1}+\frac{\Delta d_i}{2t}}} (V-d_j-{S_j}-tx)dx~,\]  
and that from secondary users is 
\[\theta \int_{\frac{0.5}{M-1}+\frac{\Delta d_i}{2t}}^{\tau^S_j(d_j)} (V-d_j-\frac{S_j}{\theta}-tx)dx~.\] 
For all users $x\in[\frac{j-1}{M-1},  \frac{j}{M-1}]$, the sum of four such expressions can be simplified to
\begin{align*}
&(1-\theta)\left(\tfrac{1}{M-1}(V-\tfrac{d_j+d_{j+1}}{2}-\tfrac{t}{2(M-1)})+\tfrac{(d_{j+1}-d_j)^2}{4t}\right) \notag\\
&\qquad +\theta \frac{t}{2}(\tau^S_j(d_j)^2+\tau^S_{j+1}(d_{j+1})^2)~.
\end{align*}

It is straightforward to show that the above expression is decreasing in $d_j$ and $d_{j+1}$ over $(0,d_0]$. Therefore, having a fast lane for either CP will increase welfare. Further, if $d_j$ increases while $d_{j+1}$ decreases (i.e., $CP_{j+1}$ prioritized while $CP_j$ throttled), the overall welfare could decrease.  

\section{Proof of Theorem \ref{thm:MCPs-single}}
We begin by finding users' consumption patterns. As $\theta=0$, we have $\tau_j^S\rightarrow -\infty$ for all CPs, and therefore we are in either cases 1 or 2 of Lemma~\ref{lemma:content-choices-users}. Further, as the condition of part (1) of Assumption~\ref{as:parameters} is satisfied, we are in case 2. This means that for users with $\frac{j-1}{M-1}\leq x\leq \frac{j}{M-1}$ (that is, those who are primarily interested in $CP_j$ and $CP_{j+1}$'s content), the content consumption choices are as follows: 
\begin{itemize}
    \item Users with $\frac{j-1}{M-1}<x<\frac{j-0.5}{M-1}+\frac{d_{j+1}-d_j}{2t}+\frac{S_{j+1}-S_j}{2t}$ single purchase on $CP_j$, 
    \item Users with $\frac{j-0.5}{M-1}+\frac{d_{j+1}-d_j}{2t}+\frac{S_{j+1}-S_j}{2t}<x<\frac{j}{M-1}$ single purchase on $CP_{j+1}$. 
\end{itemize}

\paragraph*{CPs' prioritization decisions}
Given users' consumption choices, we observe that if a $CP_j$ lowers his offered delay $d_{j}$, then he will serve a larger mass of the user market. 
We evaluate the changes in a CP's revenue when opting for prioritization (accounting for the change in his users), so as to determine the incentive of the CP for paying prioritization fees. 

We first consider the revenue each $CP_j$ can attain from the users he is attracting under competition with one of his neighbors, $CP_{j+1}$. The portion of $CP_j$ revenue that is accrued through this competition is therefore given by 
\begin{align}
\bar{R}_{j,j+1}(d_j)&= (\frac{0.5}{M-1}+\frac{d_{j+1}-d_j}{2t}+\frac{S_{j+1}-S_j}{2t})\notag\\
& \hspace{0.25in}(S_j+\lambda r_j -\lambda z_j p_j(d_j))
\label{eq:Rj-Mcp-j+1-single}
\end{align}
Note that the revenue when $CP_j$ does not opt for prioritization can be obtained by setting $d_j=d_0$ and $z_j=0$ in \eqref{eq:Rj-Mcp-j+1-single}. Similarly, $CP_j$'s portion of revenue accrued from competition with $CP_{j-1}$, $\bar{R}_{j,j-1}(d_j)$, can be derived analogously by swapping the indices $j+1$ with $j-1$ in \eqref{eq:Rj-Mcp-j+1-single}. Therefore, the total revenue to a middle-CP is given by 
\begin{align}
    R_j(d_j) &=  (\frac{1}{M-1}+\tfrac{d_{j+1}+d_{j-1}-2d_j}{2t}+\tfrac{S_{j+1}+S_{j-1}-2S_j}{2t})\notag\\
    &(S_j+\lambda r_j -\lambda z_j p_j(d_j))~.
    \label{eq:Rj-Mcp-new}
\end{align}
For the two end-CPs, $CP_1$ and $CP_M$, who only have one competitor, the total revenue is given by $R_1(d_1)=\bar{R}_{1,2}(d_1)$ and $R_M(d_M)=\bar{R}_{M,M-1}(d_M)$, respectively. 

Then, $CP_j$ will have an incentive to pay the required prioritization price if the CP's revenue given in \eqref{eq:Rj-Mcp-new}
increases following prioritization compared to his no prioritization revenue at $d_0$. Let $n_j(d_j,d_{j\pm 1}):=\frac{1}{M-1}+\frac{d_{j+1}+d_{j-1}-2d_j}{2t}+\frac{S_{j+1}+S_{j-1}-2S_j}{2t})$ for mid-CPs, and $n_j(d_j,d_{j\pm 1}):=\frac{0.5}{M-1}+\frac{d_{j\pm 1}-2d_j}{2t}+\frac{S_{j\pm 1}-S_j}{2t})$ for end CPs. Then, the price charged to $CP_j$ for delay $d_j$ should satisfy the following constraint: 
\begin{align}
\lambda n_j(d_j,d_{j\pm 1})p_j(d_{j}) &\leq \frac{d_0-d_j}{t}(S_j+\lambda r_j)~.
\label{eq:cp-ir-new}
\end{align}
Note that in contrast to the dual-purchasing scenario, this price depends on the prioritization deals offered to $CP_j$'s competitors. 

\paragraph*{The ISP's choices}
Lastly, the ISP's profit is given by
\begin{align*}
    \Pi &= F + \sum_{j=1}^{M} \lambda z_{j} n_j(d_j, d_{j\pm 1})(p_j(d_j) - C(d_{j}))~.
\end{align*}
where $C(d_j)$ is the per-user, per-unit of traffic, cost for building additional infrastructure to provide a delay $d_{j}$ to the $n_j(d_j, d_{j\pm 1})$ users of a prioritized provider $CP_j$. 

We first note that as the ISP's profit is increasing in the prioritization price of any CP, she will charge the maximum possible price the CP is willing to pay. That is, the price charged to $CP_j$ will bind \eqref{eq:cp-ir-new}.

The ISP further needs to determine her proposed QoS/delay $d$s to all CPs simultaneously. Substituting for the price binding \eqref{eq:cp-ir-new}, this will be given by the solution to
\begin{align*}
    \max_{d, ~d_{j}\in(0, d_0]} \sum_{j=1}^{M} \big(\frac{d_0-d_j}{t}(S_j+\lambda r_j) + \lambda n_j(d_j, d_{j\pm 1}) C(d_{j}))\big)~.
\end{align*}

\subsection{Analysis under relaxations of Assumption \ref{as:parameters}}\label{app:assumption1-relaxed}

In this appendix, we discuss how several of our results continue to hold qualitatively under relaxations of Assumption \ref{as:parameters}. 

\paragraph{Small $\theta$: no dual-purchasing} We begin by considering the violations of the lower bound set on $\theta$ in Part (2). When $\theta$ is significantly small, no dual-purchasing behavior will exist at equilibrium. The analysis of this case will lead to the single-purchasing equilibrium outlined in Theorem~\ref{thm:MCPs-single}. In particular, this happens if $\theta < \frac{S_j}{V-\frac12 \frac{t}{M-1}-\frac12(d_{j+1}+d_j)}, \forall j$, and correspond to the user choices outlined in cases 1 and 2 of Lemma~\ref{lemma:content-choices-users}. As discussed in Section~\ref{sec:disc-MCPs}, the lack of dual-purchasing means that all CPs (regardless of size and revenue model) will be negatively impacted by prioritization of their competitors. 

\paragraph{Large $\theta$: no single-purchasing} We next consider the violation of the upper bound set on $\theta$ in Part (2). We consider a residual benefit rate such that $\theta>\frac{S_j}{V-d_j-\frac{t}{M-1}}, \forall j$. Then, all users will be dual-purchasing content, with or without prioritization. In particular, each middle CP (end CP) will have a total of $n_j=2\frac{1}{M-1}$ ($n_j=\frac{1}{M-1}$) users and all these users will be dual-purchasing. The number of each CP's primary providers will be given by the expressions in Theorem~\ref{thm:MCPs}. Our results will continue to hold qualitatively in this scenario. That is, if a CP's competitors are prioritized, the non-prioritized CP will not lose users, but will lose primary users by the same amount outlined in Corollary~\ref{cor:cp-users-change}, and have a revenue change as outlined in Corollary~\ref{cor:cp-revenue}.

Finally, we consider relaxing the full coverage condition of Part (1), by assuming $V<d_0+\frac12\frac{t}{M-1}+\min_j S_j+F$. In this case, there exist users who do not have an incentive to purchase from any of the content providers in the non-prioritized equilibrium. By introducing prioritization, a significant decrease in $d_k, k\in\arg\min_j S_j$ can incentivize some users to adopt $CP_k$ as their primary provider. The changes in existing users' preferences for content will follow the same arguments as those of Theorem~\ref{thm:MCPs}. This means that, consistent with our original findings, a prioritized CP will gain additional users, while a non-prioritized CP will not lose any users. 
\subsection{Market equilibrium under a non-uniform user distribution}\label{app:thm3-proof-inpaper}

In this appendix, we extend our analysis to non-uniform distributions of users in the Hotelling model. This allows us to account for the popularity of certain CPs even prior to prioritization. Let $F(x)$ denote the CDF of the distribution of users in the Hotelling model.  Then, the market equilibrium is characterized by the following theorem. 

\begin{theorem}\label{thm:MCPs-nonuniform-inpaper}
Assume the conditions in Assumption \ref{as:parameters} are satisfied. Let $\tau_j(d_j):=\frac{1}{t}\left(V-d_j-\frac{S_j}{\theta}\right)$. Then, 
\begin{enumerate}
\item A $CP_j$, $j\in[2,M-1]$, with QoS/delay $d_j$ will have a total of
\begin{align*}
n_j &= F(\frac{j-1}{M-1}+\tau_j(d_j)) - F(\frac{j-1}{M-1}-\tau_j(d_j))~,
\end{align*}
users, with
\begin{align*}
    n_{j1} &= F(\frac{j-\frac12}{M-1}+\frac{d_{j+1}-d_j}{2t}) - F(\frac{j-\frac32}{M-1}+\frac{d_{j}-d_{j-1}}{2t})~,
\end{align*} of them being primary users, and the remainder as secondary users. 
\item $CP_1$ and $CP_M$ will have a total of
\begin{align*}
n_1 &= F(\tau_1(d_1))~,\\
n_M &= 1 - F(1-\tau_{M}(d_{M}))
\end{align*}
users, with
\begin{align*}
    n_{11} &= F(\frac{\frac12}{M-1}+\frac{d_{2}-d_1}{2t})~,\\
    n_{M1} &= 1 -  F(\frac{M-\frac32}{M-1}+\frac{d_{M}-d_{M-1}}{2t})~,
\end{align*} 
of them being primary users, and the remainder as secondary users, respectively. 
\end{enumerate}
Further, in a prioritized regime, letting $\bar{R}(d_j,d_{j\pm 1})=n_j(d_j)S_j + \lambda(r_j(n_{j1}(d_j,d_{j\pm 1})+\delta n_{j2}(d_j,d_{j\pm 1}))$:
\begin{enumerate}
\setcounter{enumi}{2}
    \item The QoS offered to $CP_j$'s users in the fast lane will be determined by
    \begin{align*}
        d^* :=&   \arg\max_{d, d_{j}\in(0, d_0]} ~
        \sum_{j=1}^{M} \bar{R}(d_j,d_{j\pm 1}) - \bar{R}(d_0,d_{j\pm 1}) \\
        & \hspace{1.8in}- \lambda n_j(d_j)C(d_{j})~.
    \end{align*}
    \item $CP_j$ will be charged a prioritization fee of
    \begin{align*}
    p_j(d_j^*,d_{j\pm 1}) =  \frac{\bar{R}_j(d^*_j,d^*_{j\pm 1})-\bar{R}_j(d_0,d^*_{j\pm 1})}{\lambda n_j(d^*_j)}~. 
    \end{align*}
    \end{enumerate}
\end{theorem}

\textbf{Proof:} Our proof steps are largely similar to that of Theorem~\ref{thm:MCPs}, with the difference that the number of users attained by each CP is now dependent on the user distribution CDF. 

\paragraph*{Users' content consumption decisions}
We begin by finding users' consumption patterns. Take users with $\frac{j-1}{M-1}\leq x\leq \frac{j}{M-1}$ (that is, those who are primarily interested in $CP_j$ and $CP_{j+1}$'s content). These users will have four potential consumption patterns: $\{j\}$, $\{j+1\}$, $\{j,j+1\}$ and $\{j+1,j\}$, where $\{j,-j\}$ (resp. $\{j\}$) is a bundle in which the user first (resp. only) consumes content from $CP_j$. 

Let $\tau^S_{j}(d_j) = \frac{1}{t}(V-d_j-\frac{S_j}{\theta})$. As the cutoff decision points identified in Lemma~\ref{lemma:content-choices-users} are independent of users' distributions, the consumption choices of these users will again be given by
\begin{itemize}
    \item Users with $\frac{j-1}{M-1}<x<\frac{j}{M-1}-\tau^S_{j+1}(d_{j+1})$ single purchase on $CP_j$, leading to $F(\frac{j}{M-1}-\tau_{j+1}(d_{j+1})) - F(\frac{j-1}{M-1})$ such users, 
    \item Users with $\frac{j}{M-1}-\tau^S_{j+1}(d_{j+1})<x<\frac{j-\frac12}{M-1}+\frac{d_{j+1}-d_j}{2t}$ dual purchase with primary $CP_j$, leading to $F(\frac{j-\frac12}{M-1}+\frac{d_{j+1}-d_j}{2t})-F(\frac{j}{M-1}-\tau^S_{j+1}(d_{j+1}))$ such users, 
    \item Users with $\frac{j-\frac12}{M-1}+\frac{d_{j+1}-d_j}{2t}<x<\frac{j-1}{M-1}+\tau^S_j(d_j)$ dual purchase with primary $CP_{j+1}$, leading to $F(\frac{j-1}{M-1}+\tau^S_j(d_j))-F(\frac{j-\frac12}{M-1}+\frac{d_{j+1}-d_j}{2t})$ such users,
    \item Users with $\frac{j-1}{M-1}+\tau^S_j(d_j)<x<\frac{j}{M-1}$ single purchase on $CP_{j+1}$, leading to $F(\frac{j}{M-1})-F(\frac{j-1}{M-1}+\tau^S_j(d_j))$ such users. 
\end{itemize}

To find the mass of primary and secondary users of each CP, we should integrate the above regions over the cdf of users' distribution in the Hotelling model, for all regions in which the CPs have competitors. This leads to total and primary users: %
\begin{align*}
n_j &= F(\frac{j-1}{M-1}+\tau_j(d_j)) - F(\frac{j-1}{M-1}-\tau_j(d_j))~,\\
n_{j1} &= F(\frac{j-\frac12}{M-1}+\frac{d_{j+1}-d_j}{2t}) - F(\frac{j-\frac32}{M-1}+\frac{d_{j}-d_{j-1}}{2t})~,
\end{align*}
for a non-end $CP_j$. The users for the end CPs can be found similarly. 

\paragraph*{CPs' prioritization decisions}
Given users' consumption choices, we observe that if a $CP_j$ lowers his offered delay $d_{j}$, then he will serve a larger mass of the user market. This includes changes in both primary and secondary users. 
We evaluate the changes in a CP's revenue when opting for prioritization (accounting for the change in his users), so as to determine the incentive of the CP for paying prioritization fees. 

We know that the net revenue of $CP_j$ (without prioritization fees) is given by
\begin{align}
\bar{R}(d_j,d_{j\pm 1})&:=n_j(d_j)S_j + \lambda(r_j(n_{j1}(d_j,d_{j\pm 1})+\delta n_{j2}(d_j,d_{j\pm 1}))~.
\label{eq:Rj-nonuniform}
\end{align}
Note that the revenue when $CP_j$ does not opt for prioritization can be obtained by setting $d_j=d_0$ and $z_j=0$ in \eqref{eq:Rj-nonuniform}. 

Then, $CP_j$ will have an incentive to pay the required prioritization price if the CP's revenue given in \eqref{eq:Rj-nonuniform} 
increases following prioritization compared to his no prioritization revenue at $d_0$. In particular, the prioritization price for delay $d_j$ should be such that $\bar{R}_j(d_j,d_{j\pm 1})-\lambda n_j(d_j)p_j(d_j)\geq \bar{R}_j(d_0,d_{j\pm 1})$, or equivalently, that
\begin{align}
p_j(d_{j}) &\leq \frac{\bar{R}_j(d_j,d_{j\pm 1})-\bar{R}_j(d_0,d_{j\pm 1})}{\lambda n_j(d_j)}
\label{eq:cp-ir-nonuniform}
\end{align}

\paragraph*{The ISP's choices}
Lastly, the ISP's profit is given by
\begin{align*}
    \Pi &= F + \sum_{j=1}^{M} \lambda z_{j} n_{j}(d_j)(p_j(d_j) - C(d_{j}))~.
\end{align*}
where $C(d_j)$ is the per-user, per-unit of traffic, cost for building additional infrastructure to provide a delay $d_{j}$ to the $n_j(d_j)$ users of a prioritized provider $CP_j$. 

We first note that as the ISP's profit is increasing in the prioritization price of any CP, she will charge the maximum possible price the CP is willing to pay. That is, the price charged to $CP_j$ will bind \eqref{eq:cp-ir-nonuniform}. The ISP will then jointly maximize the delays offered to all CPs, as the price charged to each CP depends on that charged to the CPs' competitors. 

If the prioritization arrangements are exclusive, the ISP will make the offer to the CP(s) from whom she can extract the most profit. \hfill \qedsymbol

\vspace{0.1in}

\paragraph*{Comparison with findings under uniform user distributions (Theorem~\ref{thm:MCPs})} From the above theorem, we observe that as CPs' total number of users $n_j$ depends only on their own  $d_j$, similar to the uniform case, no CP will lose users when his competitors are prioritized. However, a non-prioritized $CP_j$ will again lose primary users $n_{1j}$ to a prioritized competitor; for instance, when $d_{j+1}$ is lowered, $F(\frac{j-\frac12}{M-1}+\frac{d_{j+1}-d_j}{2t})$ decreases, lowering the primary users of $CP_j$. We further observe that this decrease is sharper for CPs that are originally more ``popular'', in that more users are intrinsically interested in their content, as the cdf of users' distribution $f(\cdot)$ (hence, the rate of change of $F(\cdot)$) is higher around these CPs. Conversely, CPs that are competing with more popular content providers stand to gain more users through paid prioritization. However, in contrast to Theorem~\ref{thm:MCPs}, the prioritization prices offered to each CP will depend on those offered to their competitors.

\end{document}